\documentclass[a4paper]{amsart}

\usepackage{style, appendix}
\begin{document}

\title{An optimal representation for the trace zero subgroup}

\author{Elisa Gorla}
\address{Elisa Gorla, Institut de math\'ematiques, Universit\'e de Neuch\^atel, Rue Emile-Argand 11, 2000 Neuch\^atel, Switzerland}
\email{\href{mailto:elisa.gorla@unine.ch}{elisa.gorla@unine.ch}}
\thanks{The authors were partially supported by the Swiss National Science Foundation under grants no.\ 123393, 150207, and 151884.}

\author{Maike Massierer}
\address{Maike Massierer, School of Mathematics and Statistics, University of New South Wales, Sydney NSW 2052, Australia}
\email{\href{mailto:maike@unsw.edu.au}{maike@unsw.edu.au}}


\begin{abstract}
We give an optimal-size representation for the elements of the trace zero subgroup of the Picard group of an elliptic or hyperelliptic curve of any genus, with respect to a field extension of any prime degree. The representation is via the coefficients of a rational function, and it is compatible with scalar multiplication of points. We provide efficient compression and decompression algorithms, and complement them with implementation results. We discuss in detail the practically relevant cases of small genus and extension degree, and compare with the other known compression methods.
\end{abstract}

\maketitle

\section{Introduction} \label{sec:intro}

Public key cryptography provides methods for secure digital communication. Among all public key cryptosystems, a relevant role is played by those based on the discrete logarithm problem (DLP). Such cryptographic systems work in finite groups which must satisfy three basic requirements: Computing the group operation must be efficient, the DLP must be hard, and there must be a convenient and compact representation for the elements. 

One such group is the trace zero subgroup of the Picard group of an elliptic or hyperelliptic curve. Given a curve defined over a finite field $\Fq$ and a field extension $\Fqn|\Fq$ of prime degree $n$, the trace zero subgroup consists of all $\Fqn$-rational divisor classes of trace zero.
While it has long been established that the trace zero subgroup provides efficient arithmetic and good security properties, an efficient representation was only known for special parameters. We bridge this gap by proposing an optimal-size representation for the elements of trace zero subgroups associated to elliptic curves and hyperelliptic curves of any genus, with respect to field extensions of any prime extension degree.

The trace zero subgroup can be realized as the $\Fq$-rational points of the trace zero variety, an abelian variety built by Weil restriction from the original curve. It was first proposed in the context of cryptography by Frey \cite{frey-99} and further studied by Naumann \cite{naumann}, Weimerskirch \cite{weimerskirch}, Blady \cite{blady}, Lange \cite{lange-phd,lange-04}, Rubin--Silverberg \cite{rubin-silverberg-02,rubin-silverberg-09}, Silverberg \cite{silverberg-05}, Avanzi--Cesena \cite{avanzi-cesena-07}, Cesena \cite{cesena-06,cesena-10}, and Diem--Scholten \cite{diem-scholten}, among others. Although the trace zero subgroup is a proper subgroup of the $\Fqn$-rational points of the Jacobian of the curve, it can be shown that solving the DLP in the Jacobian can be reduced to solving the DLP in the trace zero subgroup. Therefore, trace zero cryptosystems may be regarded as the (hyper)elliptic curve analog of torus-based cryptosystems such as LUC \cite{luc}, Gong--Harn \cite{gong-harn-99}, XTR \cite{xtr}, and CEILIDIH \cite{rubin-silverberg-03}. 

The trace zero subgroup is of interest in the context of pairing-based cryptography. Rubin and Silverberg have shown in \cite{rubin-silverberg-02,rubin-silverberg-09} that the security of pairing-based cryptosystems can be improved by using abelian varieties of dimension greater than one in place of elliptic curves. Jacobians of hyperelliptic curves and trace zero varieties are prominent examples for such applications. 

Scalar multiplication in the trace zero subgroup is particularly efficient, due to a speed-up using the Frobenius endomorphism, see \cite{lange-phd, lange-04,avanzi-cesena-07}. This technique is similar to the one used on Koblitz curves \cite{koblitz-91} and has been afterwards applied to GLV/GLS curves \cite{glv,gls}, which are the basis for several recent implementation speed records for elliptic curve arithmetic \cite{longa-sica-12,faz-hernandez-13,bos-costello-hisil-lauter-13}. In~\cite{avanzi-cesena-07}, Avanzi and Cesena show that trace zero subgroups often deliver better scalar multiplication performance than elliptic curves. 
E.g., scalar multiplication in trace zero subgroups of elliptic curves over a degree 5 extension field is almost 3 times faster than in elliptic curves, for the same group size. 

Since the trace zero subgroup is a subgroup of the Picard group, one may represent its elements in the same way as one represents the elements of the Picard group. Such a representation, however, sacrifices memory and bandwidth. In this paper, we solve this problem by providing a representation for the elements of trace zero subgroups which is both efficiently computable and optimal in size.
Since the trace zero subgroup has about $q^{(n-1)g}$ elements, an optimal-size representation should consist of approximately $\log_2 q^{(n-1)g}$ bits. A natural solution would be representing an element of the trace zero subgroup via $(n-1)g$ elements of $\Fq$. Such representations have been proposed by Naumann \cite[Chapter 4.2]{naumann} for trace zero subgroups of elliptic curves and by Lange \cite{lange-04} for trace zero varieties associated to hyperelliptic curves of genus 2, both with respect to cubic field extensions, and 
by Silverberg \cite{silverberg-05} and Gorla--Massierer \cite{gorla-massierer-1} for elliptic curves with respect to base field extensions of degree 3 and 5. A compact representation for Koblitz curves has been proposed by Eagle, Galbraith, and Ong \cite{eagle-galbraith-ong-11}.

In this paper we give a new optimal-size representation for the elements of the trace zero subgroup associated to an elliptic or hyperelliptic curve of any genus $g$ and any field extension of prime degree $n$. It is conceptually different from all previous representations, and it is the first representation that works for elliptic curves with $n>5$, for hyperelliptic curves of genus 2 with $n>3$, and for hyperelliptic curves of genus $g>2$. The basic idea is to represent a given divisor class via the coefficients of the rational function whose associated principal divisor is the trace of the given divisor. Our representation enjoys convenient properties, e.g., modulo the action the Frobenius the representation is injective and scalar multiplication is well-defined. In the context of a DLP-based primitive where the only operation required is scalar multiplication of points, this enables us to compute with equivalence classes of trace zero elements modulo the action of the Frobenius, and no extra bits are required to distinguish between the different representatives. 

We also give a compression algorithm to compute the representation, and a decompression algorithm to recover the original divisor class. We show that our algorithms are comparable with or more efficient than all previously known methods, when one compares the total time required for compression and decompression.

The paper is organized as follows: In Section \ref{sec:prelim} we give some preliminaries on (hyper)elliptic curves, the trace zero variety, and optimal representations. In Section \ref{sec:rep} we discuss the representation, together with compression and decompression algorithms, and we specialize these results to elliptic curves in Section \ref{sec:ec}. In Section \ref{sec:timings} we present some implementation results, as well as a detailed comparison with the other compression methods. Finally, in the Appendix we give explicit equations for the relevant cases $g=1, n=3,5$ and $g=2,n=3$.

\subsubsection*{Acknowledgements} We thank Tanja Lange for bringing to our attention the work of Blady and Naumann, and we are grateful to the mathematics department of the University of Zurich for access to their computing facilities. 

\section{Preliminaries} \label{sec:prelim}

We start by recalling the definitions and basic facts that we will need in this paper, and fixing some notation.

\subsection{Elliptic and hyperelliptic curves}

Let $C$ be a projective elliptic or hyperelliptic curve of genus $g$ defined over a finite field $\Fq$ that has an $\Fq$-rational Weierstra\ss\ point. For ease of exposition, we assume that $\Fq$ does not have characteristic 2. By making the necessary adjustments, the content of this paper carries over to the binary case.
If $\Fq$ has odd characteristic, then $C$ can be given by an affine equation of the form $$C : y^2 = f(x)$$ with $f \in \Fq[x]$ monic of degree $2g+1$ and with no multiple zeros. We denote by $\O$ the point at infinity and by $\Div$ the group of divisors on $C$.
Let $w$ be the involution $$w : C \rightarrow C, \quad (X,Y) \mapsto (X,-Y), \quad \O \mapsto \O.$$
The Frobenius map on $C$ is defined as $$\varphi : C \rightarrow C,\quad (X,Y) \mapsto (X^q,Y^q),\quad \O \mapsto \O.$$
Both $w$ and $\varphi$ extend to group homomorphisms on $\Div$.

Let $\Fqn$ be an extension field of $\Fq$, $n\geq 1$. A divisor $D$ is {\it $\Fqn$-rational} if $\varphi^n(D) = D$. We denote by
 $\Div(\Fqn)$ the $\Fqn$-rational divisors on $C$. $\Div(\Fqn)$ is a subgroup of $\Div$. 

Let $D_1=a_1P_1+\ldots+a_kP_k-a\O,D_2=b_1P_1+\ldots+b_kP_k-b\O\in\Div$, $a_i,b_i,a,b\in\N\cup\{0\}$, be two divisors of degree zero. If $a_i\leq b_i$ for all $i$ we write $D_1\leq D_2$.

As usual in the cryptographic setting, we work in the Picard group $\pic$ of $C$. This is the group of degree zero divisor classes, modulo principal divisors. For any $D,D_1,D_2\in\Div$, we write $[D]$ for the equivalence class of $D$ in $\pic$ and $D_1 \sim D_2$ for $[D_1] = [D_2]$. 
The $\Fqn$-rational divisor class $[D]$ is the equivalence class of the $\Fqn$-rational divisor $D$. The subgroup of $\pic$ consisting of the $\Fqn$-rational divisor classes is denoted by $\pic(\Fqn)$.

A divisor $D = P_1 + \ldots + P_r - r\O\in\Div^0$ is {\it semi-reduced} if $P_i \in C \setminus \{\O\}$ and $P_i \neq w(P_j)$ for $i \neq j$.
$D$ is {\em reduced} if it is semi-reduced and in addition $r \in \{0,\ldots,g\}$.
Notice that $D$ is reduced with $r=0$ if and only if $[D] = 0$. 

It follows from the Riemann--Roch Theorem that every degree zero divisor class can be represented by a unique reduced divisor.
For any divisors $D_1, D_2\in\Div^0$, we denote by $D_1 \oplus D_2$ the reduced divisor such that $[D_1 \oplus D_2]=[D_1 + D_2]$. 
When $C$ is an elliptic curve, then each non-zero element of $\pic$ is uniquely represented by a divisor of the form $P - \O$ with $P \in C$. In fact, we have $C \cong \pic$ as groups via $P \mapsto [P - \O]$. 
For elliptic curves, we denote a divisor class by the unique corresponding $P\in C$. In particular, we denote $0\in\pic$ by  the point $\O$. 

There is a one-to-one correspondence between semi-reduced divisors $D = P_1 + \ldots + P_r - r\O$ and pairs of polynomials $(u,v)$ such that $u$ is monic, $\deg v < \deg u$, and $u \mid v^2-f$: Given a divisor $D$, then $u(x) = \prod_{i=1}^{r} (x-X_i)$ and $v(x)$ is the unique polynomial such that $v(X_i) = Y_i$ with multiplicity equal to the multiplicity of $P_i$ in $D$. The polynomial $v(x)$ may be computed by solving a linear system. Conversely, given polynomials $u,v$ as above, let $D=\Delta-\deg(\Delta)\O$ where $\Delta$ is the effective divisor with defining ideal $I_\Delta=(u(x),y-v(x)).$  It is easy to show that $D$ is semi-reduced. Notice that since $u\mid v^2-f$, then $y^2-f\in(u,y-v)$.
The correspondence restricts to a correspondence between reduced divisors and pairs of polynomials $(u,v)$ such that $u$ is monic, $\deg v < \deg u\leq g$, and $u \mid v^2-f$.

A commonly used representation for divisor classes is the {\it Mumford representation}. An element $[D] \in \pic$ with $D$ 
a reduced divisor is represented by the pair of polynomials $[u(x),v(x)]$ associated to it in the correspondence described in the previous paragraph. The Mumford representation is particularly useful when computing with divisor classes, and all algorithms given in this paper make use of this representation. 
If $C$ is an elliptic curve, then the Mumford representation of $P = (X,Y) \in C$ is $[x-X,Y]$. It follows from the definition that the Mumford representation of $[0]$ is $[1,0]$.
A convenient property of the Mumford representation is that $\Fqn$-rationality of divisor classes is easily detected: $[u,v] \in \pic(\Fqn)$ if and only if $u, v \in \Fqn[x]$. 

By definition, a reduced divisor $D\in\Div(\Fqn)$ with $[D]=[u,v]$ is {\it prime} if $u\in\Fqn[x]$ is an irreducible polynomial. This is equivalent to the statement that $(u,y-v)$ is a prime ideal of $\Fqn[x,y]/(y^2-f(x))$. Notice that being prime depends on the choice of $\Fqn$. Sometimes we write a divisor as a sum of prime divisors: $D=D_1+\ldots+D_t$, with $D_i\in\Div(\Fqn)$ prime. The prime divisors $D_1,\ldots,D_t$ are unique up to permutation, but not necessarily distinct. If $[D_i]=[u_i,v_i]$ is the Mumford representation, then $u=\prod_{i=1}^t u_i$ is the irreducible factorization of $u\in \Fqn[x]$.

Cantor's Algorithm performs the addition of divisor classes in the Mumford representation. 
For elliptic curves and hyperelliptic curves of genus 2, there exist explicit addition formulas that are easier to use and more efficient than Cantor's Algorithm (see \cite{washington} and \cite{lange-03}).

\subsection{The trace zero variety and optimal representations}

The trace endomorphism in the divisor group of $C$ with respect to the extension $\Fqn|\Fq$ is defined by
$$ \Tr : \Div(\Fqn) \rightarrow \Div(\Fq), \quad D \mapsto D + \varphi(D) + \ldots + \varphi^{n-1}(D).$$

Throughout the paper, we denote by $u^{\varphi}$ the application of the finite field Frobenius automorphism $\varphi: \Fqbar \rightarrow \Fqbar$ to the coefficients of a polynomial $u$. We denote the product $u u^{\varphi}\cdots u^{\varphi^{n-1}}$ by $u^{1+\varphi+\ldots+\varphi^{n-1}}$ or by $N(u)$, and we call it the {\em norm} of $u$.

\begin{lemma}\label{inverseimages}
The trace homomorphism $\Tr :\Div(\Fqn)\rightarrow\Div(\Fq)$ has the following properties:
\begin{enumerate}
\item For any prime divisor $D$ we have $\Tr^{-1}(\Tr(D))=\{D,\varphi(D),\ldots,\varphi^{n-1}(D)\}.$
\item $D\in\Div(\Fqn)\setminus\Div(\Fq)$ is a prime divisor if and only if $\Tr(D)\in\Div(\Fq)$ is a prime divisor.
\end{enumerate}
\end{lemma}

\begin{proof}
$(i)$ Let $D \in \Div(\Fqn)$ be a prime divisor with $[D]=[u,v]$, $u \in \Fqn[x]$ irreducible. Then $\Tr (D)$ has $u$-polynomial $N(u) = u u^{\varphi} \cdots u^{\varphi^{n-1}}$, where all the $u^{\varphi^j}$ are irreducible over $\Fqn$. Hence any $D'$ with $\Tr (D') = \Tr (D)$ has to have as $u$-polynomial one of the $u^{\varphi^j}$, and therefore $D' = \varphi^j (D)$ for some $j \in \{0,\ldots,n-1\}$. Conversely, $\Tr(\varphi^j(D))=\Tr(D)$ for all $j$.

$(ii)$ This is a restatement of the well known fact that that $u \in \Fqn[x]\setminus\Fq[x]$ is irreducible if and only if $N(u) = uu^{\varphi}\cdots u^{\varphi^{n-1}} \in \Fq[x]$ is irreducible.
\end{proof}

Since the Frobenius map is well-defined as an endomorphism on divisor classes, we also have a trace endomorphism $[\Tr]$ in the Picard group
$$[\Tr] : \pic(\Fqn) \rightarrow \pic(\Fq), \quad [D] \mapsto [D+\varphi(D)+\ldots+\varphi^{n-1}(D)]. $$
We are interested in the kernel of this map.
\begin{definition}\label{def:tzv}
 Let $n$ be a prime number. Then the {\it trace zero subgroup} of $\pic(\Fqn)$ is
 $$T_n = \{ [D] \in \pic(\Fqn) \mid \Tr(D) \sim 0 \}. $$
\end{definition}

Using Weil restriction, the points of $T_n$ can be viewed as the $\Fq$-rational points of a $g(n-1)$-dimensional variety defined over $\Fq$, called the {\it trace zero variety}. For a proof and more details, see \cite[Chapters 7.4.2 and 15.3]{handbook-hecc}.

Interest in the trace zero variety in the cryptographic context was first raised by Frey in \cite{frey-99}. The main advantages of working in $T_n$ are that
addition in the trace zero subgroup may be sped up considerably by using the Frobenius endomorphism, and that it yields high security parameters in the context of pairing-based cryptography, for some values of $n$ and $g$. Moreover, the DLP in $\pic(\Fqn)$ is as hard as the DLP in $T_n$.
Therefore, working in $T_n$ allows us to reduce the key length with respect to $\pic(\Fqn)$ without compromising the hardness of the DLP. In order to reduce the key length however, one needs to find an efficient representation for its elements. In this paper, we give an optimal one for any $g$ and any prime $n$. 

We start by showing that solving the DLP in $\pic(\Fqn)$ can be reduced to solving the DLP in $T_n$.

\begin{proposition} \label{prop:exsequence}
We have a short exact sequence
$$ 0 \longrightarrow \pic(\Fq) \longrightarrow \pic(\Fqn) \overset{[\varphi-\id]}{\longrightarrow} T_n \longrightarrow 0. $$
In particular, solving a DLP in $\pic(\Fqn)$ has the same complexity as solving a DLP in $T_n$ and a DLP in $\pic(\Fq)$.
\end{proposition}

\begin{proof}
Surjectivity of $[\varphi-\id]$ holds according to \cite[Proposition 7.13]{handbook-hecc}. This proves that we have a short exact sequence as claimed.
By the standard reduction obtained by combining an effective version of the Chinese Remainder Theorem and the Pohlig--Hellman Algorithm, we may assume without loss of generality that we are solving a DLP of the form $a[D]=[D']$, where $[D],[D']\in\pic(\Fqn)$ and $[D]$ has prime order. If $[\varphi(D)-D]\neq 0$, then $[\varphi(D)-D]$ and $[D]$ have the same order, and the DLP may be mapped to $T_n$ via $[\varphi-\id]$ and solved there. Else, $[D]\in\pic(\Fq)$.
\end{proof}

\begin{remark}
We stress that the choice of good parameters is crucial for the security of trace zero cryptosystems. While Lange \cite{lange-04}, Avanzi--Cesena \cite{avanzi-cesena-07}, and Rubin--Silverberg \cite{rubin-silverberg-09} have shown that for certain choices of $n$ and $g$ trace zero subgroups are useful and secure in the context of pairing-based cryptography, there may be security issues in connection with DLP-based cryptosystems. For example, Weil descent attacks (see \cite{ghs, diem-ghs,diem-scholten}) and index calculus attacks (see \cite{gaudry-09,enge-gaudry-thome-11,diem-11}) may apply. However, Weil descent attacks only apply to a very small proportion of all curves, and index calculus attacks often have large constants hidden in the asymptotic complexity analysis, thus making them very hard to realize in practice. Nevertheless, special care must be taken to choose good parameters and avoid weak curves. E.g., for $g=1$ and $n=3$ and for most curves, computing a DLP in the trace zero subgroup has square root complexity. For a more complete discussion of the complexity of DLP algorithms for the trace zero subgroup, see also ~\cite{gorla-massierer-3}.
\end{remark}

\begin{remark} \label{rmk:card}
As a consequence of the exact sequence in Proposition \ref{prop:exsequence} we obtain that the cardinality of the trace zero subgroup may be computed easily in terms of the coefficients of the characteristic polynomial of Frobenius, see also \cite[Chapter 15.3.1]{handbook-hecc}. 
In particular, counting the number of points in $T_n$ only requires determining the characteristic polynomial of a curve defined over $\Fq$. Counting the number of points of an elliptic or hyperelliptic curve of, e.g., the same genus and comparable group size would require determining the characteristic polynomial of a curve defined over $\F_{q^{n-1}}$.
\end{remark}

The question of finding an optimal-size representation for the elements of the trace zero subgroup has been investigated in previous works both for elliptic and hyperelliptic curves, and it is stated as an open problem in the conclusions of \cite{avanzi-cesena-07}. The analogous problem for primitive subgroups of finite fields leads to torus-based cryptography, which was introduced by Rubin and Silverberg in~\cite{rubin-silverberg-03}. 

\begin{definition}\label{repr}
Let $A$ be a $d$ dimensional abelian variety defined over $\Fq$. 
A {\it representation} for the elements of $A(\Fq)$ is a map $$\mathcal{R} : A(\Fq) \longrightarrow \mathbb{F}_q^{\ell}\times\mathbb{F}_2^{k}.$$
\end{definition}

Notice that, in our setup, a representation map $\mathcal{R}$ is not necessarily injective. Nevertheless, any representation induces an injective representation $$\overline{\mathcal{R}}: A(\Fq)/{\sim} \longrightarrow \mathbb{F}_q^\ell\times\mathbb{F}_2^{k},$$ where $P\sim Q$ iff $\mathcal{R}(P)=\mathcal{R}(Q)$ for any $P,Q\in A(\Fq)$. 
Sometimes we do not distinguish between $\mathcal{R}$ and $\overline{\mathcal{R}}$, and say that $x\in\im\mathcal{R}$ is a representation for the {\it class} $\mathcal{R}^{-1}(x)$. 

\begin{definition}\label{opt}
Let $d\geq 1$ be an integer. Let $\A$ be a set of pairs $(A,\Fq)$, where $A$ is a $d$ dimensional abelian variety defined over $\Fq$ with at least one $\Fq$-rational point. 
An {\it optimal representation} for $\A$ is a family of representations
$$\mathcal{R}: A(\mathbb{F}_{q}) \longrightarrow \mathbb{F}_{q}^d\times\mathbb{F}_2^k$$ for all $(A,\Fq)\in\A$, with the property that $k$ and the cardinality of $\mathcal{R}^{-1}(x)$ are upper bounded by constants which do not depend on $(A,\Fq)\in\A$.
We also say that each map $$\mathcal{R}: A(\mathbb{F}_{q}) \longrightarrow \mathbb{F}_{q}^d\times\mathbb{F}_2^k$$ 
is an {\it optimal representation} for the elements of $A(\Fq)$.

Given $P\in A(\Fq)$, $x\in\im\mathcal{R}$, we refer to computing $\mathcal{R}(P)$ as {\it compression} and $\mathcal{R}^{-1}(x)$ as {\it decompression}.
\end{definition}

It was shown in~\cite{lang-weil-54} that for any abelian variety $A$ defined over $\Fq$ one has $$|A(\Fq)|=q^d+ O(q^{d-\frac{1}{2}}).$$ Hence, intuitively, a representation $\mathcal{R}$ for $\A$ is optimal if it allows us to represent the elements of $A(\Fq)$ for every $(A,\Fq)\in \A$ with the smallest possible number of elements of $\Fq$, for $q\gg 0$. The number $k$ of extra bits is independent of $q$, hence it becomes negligible for $q\gg 0$.

\begin{remark}
Sometimes we deal with representations which are not defined on the zero element of the group. However, this is not a problem in practice, and it is in fact common in cryptographic use (as one sees in the following examples).
\end{remark}

\begin{example}\label{repr_ec}
Let $\A=\{(E,\Fq)\mid \ q \mbox{ prime power}, \ E \mbox{ elliptic curve defined over $\Fq$}\}$. Assume that the elliptic curves are in short Weierstrass form.
One has the usual representation 
$$\begin{array}{rcl} \mathcal{R} : E(\Fq)\setminus\{\O\} & \longrightarrow & \Fq\\
(X,Y) & \longmapsto & X.\end{array}$$
For any $X\in\mathcal{R}(E(\Fq))$ we have $\mathcal{R}^{-1}(X)=\{(X,Y),(X,-Y)\}$. Compression has no computational cost, and decompression is efficient, since $Y$ can be recomputed, up to sign, from the equation of the curve at the cost of computing a square root in $\Fq$. 

Appending to the image of each point an extra bit corresponding to the sign of the $y$-coordinate yields an injective map 
$$\mathcal{R}' : E(\Fq)\setminus\{\O\} \longrightarrow \Fq \times \F_2.$$ Both $\mathcal{R}$ and $\mathcal{R}'$ are optimal representations for $\A$.
\end{example}

The same logic applies to hyperelliptic curves. 

\begin{example}\label{repr_hc}
Let $g\geq 2$ be an integer. Let $$\A=\{(\pic,\Fq)\mid \ q \mbox{ prime power}, \ C \mbox{ plane hyperelliptic curve of genus $g$ defined over $\Fq$}\}.$$ 
Assume that the hyperelliptic curves have equations of the form $y^2=f(x)$, with $\deg f=2g+1$.
The following is an optimal representation proposed by Hess--Seroussi--Smart in~ \cite{hess-seroussi-smart-01}:
$$\begin{array}{ccc} \mathcal{R}:\pic(\Fq) & \longrightarrow &\F_q^g\times\F_2 \\
\mbox{$[D]=[u=\sum_{i=0}^{g} u_ix^i,v]$} & \longmapsto & (u_0,\ldots,u_{g-1},\delta)\end{array}$$
where $u_i=0$ for $i>r=\deg u$, $\delta=1$ if $r=g$, and $0$ otherwise.
The polynomial $u$ contains all the information about the $x$-coordinates of the points $P_i$ in the support of the reduced divisor $D = P_1 + \ldots + P_r - r\O$, but not about the signs of the corresponding $y$-coordinates. Therefore $\mathcal{R}$ identifies up to $2^g$ elements of $\pic(\Fq)$. 
As before, one can use $g$ extra bits to store these signs, making the representation injective (see \cite{hess-seroussi-smart-01}). 
A different optimal representation for the elements of $\pic(\Fq)$ is given by Stahlke \cite{stahlke-04}.
\end{example}

\begin{example}\label{repr_torus}
Let $n\geq 2$ be an integer. For each prime power $q$, let $P_{q,n}$ be the primitive subgroup of the multiplicative group $\mathbb{G}_m$, relative to the field extension $\mathbb{F}_q^n|\Fq$. $P_{q,n}$ is a $\phi(n)$ dimensional abelian subvariety of the Weil restriction of scalars $\Res_{\mathbb{F}_q^n|\Fq}\mathbb{G}_m$, where $\phi(n)=|\{1\leq m\leq n\mid (m,n)=1\}|$ is the Euler $\phi$ function. Let $\A_n=\{(P_{q,n},\Fq)\mid \ q \mbox{ a prime power}\}$. Finding an optimal representation for $\A_n$ is at the core of torus-based cryptography. This problem was solved for $n=2,3,6,30$ in several works, including \cite{luc,gong-harn-99,xtr,rubin-silverberg-03,rubin-silverberg-04,vandijk-woodruff-04,dgprssw-05,rubin-silverberg-08,shirase-han-hibin-kim-takagi-08,karabina-10,karabina-12,yonemura-isogai-muratani-hanatani-12}. 
\end{example}

\begin{notation}
Let $g\geq 2$ be an integer, $n$ be a prime number. Let $$\T_{n,1}=\{(T,\Fq)\mid \ q \mbox{ prime power},\ T \mbox{ trace zero variety of an elliptic curve}\}$$ and
$$\T_{n,g}=\{(T,\Fq)\mid \ q \mbox{ prime power},\ T \mbox{ trace zero variety of a hyperelliptic curve of genus $g$}\}$$ 
where all trace zero varieties are relative to a field extension of fixed degree $n$.
\end{notation}

In this paper, we construct representations for $\T_{n,g}$, $g\geq 1$, of the form $$\mathcal{R} : T_n\longrightarrow \mathbb{F}_q^{g(n-1)}\times\mathbb{F}_2$$ with the property that each element in the image has at most $n^g$ inverse images. 

\begin{remark}
Since $T_n\subset\pic(\Fqn)$, we may use the representations of Examples~\ref{repr_ec} and~\ref{repr_hc} for the family $\T_{n,g}$. However such representation are not optimal, since the dimension of the varieties in $\T_{n,g}$ is $(n-1)g$.
\end{remark}

\section{An optimal representation for the trace zero subgroup via rational functions} \label{sec:rep}

In this section, we give an optimal representation for the family $\T_{n,g}$ of trace zero varieties of elliptic curves 
or hyperelliptic curves of fixed genus $g$, with respect to a field extension of fixed degree $n$.
A simple example is the case of elliptic curves $E$ and extension degree $n=2$, where
$$T_2 = \{ (X,Y) \in E(\F_{q^2}) \mid X \in \Fq, Y \in (\F_{q^2} \setminus \Fq) \cup \{0\} \} \cup \{\O\}.$$ 
Hence the $x$-coordinate of the points of $T_2$ yields an optimal representation (see \cite[Proposition 2]{gorla-massierer-1}). 
This statement can be easily generalized to higher genus curves when $n=2$. We omit the proof, since the proposition is a special case of the next theorem.

\begin{proposition}\label{n=2}
Fix $g\geq 1$ and let $C$ be an elliptic or hyperelliptic curve of genus $g$ defined over $\Fq$. 
Let $T_2\subseteq\pic(\mathbb{F}_{q^2})$ be the trace zero subgroup corresponding to the field extension $\F_{q^2}|\Fq$. 
Let $$\begin{array}{rccl}
\mathcal{R} : & T_2 & \longrightarrow & \mathbb{F}_q^g\times\mathbb{F}_2 \\
 & \mbox{$[u,v]$} & \longmapsto & (u_0,\ldots,u_{g-1},\delta)\end{array}$$
 where $u=\sum_{i=0}^g u_ix^i$ is monic of degree $0\leq r\leq g$, $\delta=1$ if $\deg u=g$, and $\delta=0$ otherwise. 
Then $$ T_2=\{[u,v]\in\pic(\mathbb{F}_{q^2})\mid u\in\Fq[x],\; v^{\varphi}=-v\}, $$
and $\mathcal{R}$ yields an optimal representation for the family $\T_{2,g}$.
\end{proposition}

We now proceed to solve the problem in the case when $n$ is any prime.
Let $D$ be a reduced divisor. We propose to represent an element $[D]$ of $T_n$ via the rational function $h_D$ on $C$ with divisor $$ \princdiv(h_D) = \Tr(D).$$ Such a function is defined over $\Fq$ since $\Tr(D)$ is, and it is unique up to multiplication by a constant. We now establish some properties of $h_D$. In particular, we show that a normalized form of $h_D$ can be represented via $g(n-1)$ elements of $\Fq$ plus an extra bit. This gives an optimal representation for the family $\T_{n,g}$, where each map identifies at most $n^g$ divisor classes. 

\begin{theorem} \label{thm:hd}
Let $D = P_1 + \ldots + P_r - r\O$ be a reduced divisor such that $[D]=[u,v] \in T_n$, and let $h_D \in \Fq(C)$ be a function such that $\princdiv(h_D) = \Tr(D)$. Write $D=D_1+\ldots+D_t$, where $D_i$ are reduced prime divisors defined over $\Fqn$. Then:
\begin{enumerate}
\item $h_D = h_{D,1}(x) + y h_{D,2}(x)$ with $h_{D,1}, h_{D,2} \in \Fq[x]$.
\item $H_D(x) := h_{D,1}(x)^2 - f(x)h_{D,2}(x)^2 \in \Fq[x]$ has degree $rn$, and its zeros over $\Fqbar$ are exactly the $x$-coordinates of the points $\varphi^{j}(P_1), \ldots, \varphi^{j}(P_r)$ for $j = 0,\ldots,n-1$. Equivalently, $H_D=N(u)$ where $N(u)$ denotes the norm of $u$ relative to $\Fqn|\Fq$.
\item $\deg h_{D,1} \leq \lfloor \frac{nr}{2} \rfloor$ and $\deg h_{D,2} \leq \lfloor \frac{nr-2g-1}{2} \rfloor$, where equality holds for the degree of $h_{D,1}$ if $r$ is even or $n=2$, and equality holds for the degree of $h_{D,2}$ if $r$ is odd and $n \ne 2$.
\item Let $F$ be a reduced divisor. Then $h_D=h_F\in\Fq(C)$ if and only if $F$ is of the form $F=\varphi^{j_1}(D_1)+\ldots+\varphi^{j_t}(D_t)$ for some $0\leq j_1,\ldots,j_t\leq n-1$. In particular, there are at most $n^g$ reduced divisors $F$ such that $h_F=h_D$.
\end{enumerate}
\end{theorem}

\begin{proof}
Since $[D] \in T_n$, we have $0\sim\Tr(D)\in\Div(\Fq)$. Hence there exists an $h_D \in \Fq(C)$ such that $\princdiv(h_D)=\Tr(D)$. The function $h_D$ is uniquely determined up to multiplication by a constant.
  
  $(i)$ The function $h_D$ is a polynomial, since it has its only pole at $\O$. Modulo the curve equation $y^2 = f(x)$, the polynomial $h_D\in\Fq[x,y]$ has the desired shape.
  
  $(ii)$ By definition, $h_D$ has zeros $\varphi^j(P_1),\ldots,\varphi^j(P_r), j = 0,\ldots,n-1$, and pole $nr \O$. Therefore, $h_D \circ w = h_{D,1}(x) - yh_{D,2}(x)$ has zeros $w(\varphi^j(P_1)),\ldots,w(\varphi^j(P_r)), j = 0,\ldots,n-1$ and pole $nr \O$. Since $H_D(x) = h_D (h_D \circ w)\in\Fq[x,y]/(y^2-f(x))$, then $H_D$ has precisely the zeros $\varphi^{j}(P_1), \ldots, \varphi^{j}(P_r),$ $w(\varphi^j(P_1)),\ldots,w(\varphi^{j}(P_r))$ for $j = 0,\ldots,n-1$ and the pole $2nr\O$. Therefore $H_D = N(u)$, up to multiplication by a constant.
  
  $(iii)$ From the fact that $\deg H_D=nr$ and $\deg f=2g+1$, we deduce the bounds on the degrees. If $r$ or $n$ is even, then $\lfloor \frac{nr}{2} \rfloor = \frac{nr}{2}$ and $\lfloor \frac{nr-2g-1}{2} \rfloor = \frac{nr}{2} - g - 1$. Therefore $\deg (h_{D,1}^2) \leq nr$ and $\deg (f h_{D,2}^2) \leq nr-1$, hence $\deg h_{D,1} = \frac{nr}{2}$. An analogous computation for $r$ and $n$ both odd shows that in this case $\deg h_{D,2} =\frac{nr-1}{2}-g=\left\lfloor\frac{nr-2g-1}{2}\right\rfloor$.
  
  $(iv)$ Let $F\in\Div(\Fqn)$ be a reduced divisor such that $h_F=h_D\in\Fq(C)$. Then 
$$\Tr(F)=\princdiv (h_F)=\princdiv (h_{D})=\Tr(D)\in\Div(\Fq).$$ 
Write $\Tr(D)=\Tr(D_1)+\ldots+\Tr(D_t)=\Tr(F)$, where $\Tr(D_i)\in\Div(\Fq)$ are prime divisors by Lemma~\ref{inverseimages} (ii). By Lemma~\ref{inverseimages} (i), $\Tr^{-1}(\Tr(D_i))=\{D_i,\varphi(D_i),\ldots,\varphi^{n-1}(D_i)\}$ for all $i$, hence $F=\varphi^{j_1}(D_1)+\ldots+\varphi^{j_t}(D_t)$ for some $j_1,\ldots,j_t\in\{0,\ldots,n-1\}$. The number of such $F$ is at most $n^t\leq n^g$.
\end{proof}

\begin{remark}
If $n=2$ and $[D]=[u(x),v(x)]\in T_2$, then $h_D(x,y)=u(x)$. Hence Theorem~\ref{thm:hd} recovers the optimal representation from Proposition~\ref{n=2}.
\end{remark}

\begin{remark}\label{rem:reduced}
Let $D\in\Div^0(\Fqn)$ be a reduced divisor, $D=D_1+\ldots+D_t$ with $D_i\in\Div^0(\Fqn)$ reduced prime divisors.
Notice that not all the divisors $F$ of the form $F=\varphi^{j_1}(D_1)+\ldots+\varphi^{j_t}(D_t)$ for some $j_1,\ldots,j_t\in\{0,\ldots,n-1\}$ are reduced. E.g., let $C$ be a hyperelliptic curve of genus $2$ and let $P\in C(\Fqn)\setminus C(\Fq)$ be a point. Then $\varphi(P)\neq P$ and $D=P+w(\varphi(P))-2\O$ is a reduced divisor. But a divisor $F=\varphi^{j_1}(P)+w(\varphi^{j_2}(P))-2\O$ is reduced if and only if $j_1\neq j_2$.
Because of this, when decompressing $\mathcal{R}([D])$ one needs to discard all the divisors classes $[F]\in T_n$ which have $\Tr(F)=\Tr(D)$, but $F$ is not a reduced divisor.
In our decompression algorithm, for a given $\alpha=\mathcal{R}([D])$ we recover one reduced $F\in\Div(\Fqn)$ such that $\mathcal{R}([F])=\alpha$. Such an $F$ uniquely identifies $\mathcal{R}^{-1}(\mathcal{R}([D]))$.
\end{remark}

The following corollary clarifies how Theorem~\ref{thm:hd} gives an optimal representation for $\T_{n,g}$, consisting of $(n-1)g$ elements of $\Fq$ and a bit. 
Using standard techniques, the representation may be made injective at the cost of appending $\lfloor g\log_2 n\rfloor + 1$ bits to it.

\begin{corollary}\label{cor:repr}
Let $n\geq 3$, let $0\neq D\in\Div(\Fqn)$ be a reduced divisor of degree zero such that $[D]=[u,v] \in T_n$, and let $r=\deg u$. 
Set $d_1=\left\lfloor\frac{ng}{2}\right\rfloor$ and $d_2=\left\lfloor\frac{(n-2)g-1}{2}\right\rfloor$.
Let $h_D=h_{D,1}(x)+y h_{D,2}(x) \in \Fq[x,y]$ be such that $\princdiv(h_D) = \Tr(D)$, where $h_{D,1} = \gamma_{d_1} x^{d_1} +\ldots+\gamma_1x+\gamma_0$, $h_{D,2} = \beta_{d_2}x^{d_2} + \ldots + \beta_1x+\beta_0$. Let $h_{D,1}$ be monic if $r$ is even, and $h_{D,2}$ be monic if $r$ is odd.
If $r=g$ let $\delta=1$, else let $\delta=0$. Define:
\begin{itemize}
\item If $g$ is even, then $$\begin{array}{rcl}
\mathcal{R} : T_n & \longrightarrow & \F_q^{(n-1)g}\times\F_2 \\
\mbox{$[D]$} & \longmapsto & (\beta_0,\ldots,\beta_{d_2},\gamma_0,\ldots,\gamma_{d_1-1},\delta)\\
\mbox{$[0]$} & \longmapsto & (0,\ldots,0).
\end{array}$$
\item If $g$ is odd, then $$\begin{array}{rcl}
\mathcal{R} : T_n & \longrightarrow & \F_q^{(n-1)g}\times\F_2 \\
\mbox{$[D]$} & \longmapsto & (\gamma_0,\ldots,\gamma_{d_1},\beta_0,\ldots,\beta_{d_2-1},\delta)\\
\mbox{$[0]$} & \longmapsto & (0,\ldots,0).
\end{array}$$
\end{itemize}
Then $\mathcal{R}$ yields an optimal representation for the family $\T_{n,g}$, with the property that every element of $\im\mathcal{R}$ has at most $n^g$ inverse images.
\end{corollary}

\begin{proof}
It follows from Theorem~\ref{thm:hd} (iii) that
$$\deg h_{D,1}\leq \left\lfloor\frac{rn}{2}\right\rfloor\leq d_1 \mbox{ and } \deg h_{D,2}\leq\left\lfloor\frac{nr-2g-1}{2}\right\rfloor\leq d_2,$$
hence the polynomials can be written as claimed. Moreover, if $g$ is even and $r<g$, then $$\deg h_{D,1}\leq \left\lfloor \frac{n(g-1)}{2}\right\rfloor\leq d_1-1\;\mbox{ and } \;\delta=0.$$ If $g=r$ is even, then $h_{D,1}$ is monic of degree $d_1$ and $\delta=1$.
If instead $g$ is odd and $r<g$, then $$\deg h_{D,2}\leq \left\lfloor \frac{n(g-1)-2g-1}{2}\right\rfloor\leq d_2-1\;\mbox{ and } \;\delta=0.$$ Finally, if $g=r$ is odd, then $h_{D,2}$ is monic of degree $d_2$ and $\delta=1$. Since $d_1+d_2+1=(n-1)g,$
then $\im\mathcal{R}\subseteq\mathbb{F}_q^{(n-1)g}\times\mathbb{F}_2$ in all cases. $\mathcal{R}$ is optimal since $(n-1)g\lceil\log_2 q\rceil+1=\lceil\log_2 |T_n|\rceil+O(1)$. 
Finally, the representation identifies at most $n^g$ elements by Theorem~\ref{thm:hd} (iv).
\end{proof}

\begin{remark} \label{rmk:opensubset1}
If one chooses to work only with divisors of the form $D = P_1+\ldots+P_g-g\O$, then the last bit in the representation of Corollary~\ref{cor:repr} may be dropped and we have a representation of size $(n-1)g\lceil\log_2 q\rceil$. 
Divisor classes whose reduced representative has this form constitute the majority of the elements of $T_n$. Moreover, there are cases in which the trace zero subgroup consists only of divisor classes represented by reduced divisors of this shape. This is the case e.g.\ for elliptic curves, where $r=1$ if $D\neq 0$. Moreover, Lange \cite[Theorem 2.2]{lange-04} proved that for $g=2$ and $n=3$, all nontrivial elements of $T_3$ are represented by reduced divisors with $r = 2 = g$. 
\end{remark}

In the next theorem we establish some facts that we use for our decompression algorithm.

\begin{theorem}\label{thm:hd2}
Let $[D]=[u,v] \in T_n$ with $D\in\Div^0$ a reduced divisor, and let $h_D=h_{D,1}(x)+y h_{D,2}(x) \in \Fq[x,y]$ be such that $\princdiv(h_D) = \Tr(D)$. Write $D=D_1+\ldots+D_t$, where $D_i\in\Div^0$ are reduced prime divisors defined over $\Fqn$ with Mumford representation $[D_i]=[u_i,v_i]$. Then:
\begin{enumerate}
\item $h_{D,2}\equiv 0 \bmod{u_i}$ if and only if $w(D_i) = \varphi^j(D_k)$ for some $j \in \{0,\ldots,n-1\}$ and some $k \in \{1,\ldots,t\}$.
\item Let $n \ne 2$. Then $w(D_i) = \varphi^j(D_i)$ for some $j\neq 0$ if and only if $D_i\in\pic[2](\Fq)$.
\item Let $n\neq 2$, $\ell,m \geq 0$, and assume that $D_i \ne w(D_i)$. Then $\Tr(D) = m \Tr (D_i) +\ell \Tr(w(D_i)) + \Tr(G)$ for some $G\in\Div^0$, where $\Tr(D_i),\Tr(w(D_i)) \not\leq \Tr(G)$ and $G$ has poles only at $\O$, if and only if $N(u_i)^{\min\{\ell,m\}}$ exactly divides $h_D$.
\end{enumerate}
\end{theorem}

\begin{proof}
$(i)$ We have $h_{D,2}(x)\equiv 0\bmod{u_i}$ if and only if $h_D(x,y)\equiv h_{D,1}(x) \equiv h_{w(D)}(x,y) \bmod{u_i}$. Since $D_i\leq \Tr(D)$, this is also equivalent to $w(D_i)\leq\Tr(D)$. Since $D_i$ is prime, $w(D_i)$ is also prime and $w(D_i)\leq\Tr(D)$
if and only if $w(D_i)=\varphi^j(D_k)$ for some $j\in\{0,\ldots,n-1\}$ and some $k\in\{1,\ldots,t\}$ by Lemma~\ref{inverseimages} (i). 

$(ii)$ We only prove the nontrivial implication. If $w(D_i) = \varphi^j(D_i)$ for some $j\neq 0$, then $u_i \in \Fq[x]$ and 
$-\nu = \nu^{\varphi^j}$ for all coefficients $\nu$ of $v_i$. Hence $\nu^2 = (\nu^2)^{\varphi^j}$, so $\nu \in \F_{q^{2j}} \cap \Fqn = \Fq$. 
Therefore also $v_i \in \Fq[x]$, hence $ w(D_i)=\varphi^j(D_i)=D_i\in\pic(\Fq)$.

$(iii)$ Let $\Tr(D) =m\Tr (D_i) + \ell \Tr(w(D_i)) + \Tr(G)$ for some divisor $G\in\Div^0$, with poles only at $\O$ and $\Tr(D_i),$ $\Tr(w(D_i)) \not\leq \Tr(G)$. Assume that $m \geq \ell$, since the proof of the other case is similar. Then
$$\princdiv(N(u_i)^{\ell} h_{D_i}^{m-\ell}h_G)=\ell \Tr(D_i) + \ell \Tr(w(D_i)) + (m-\ell)\Tr(D_i) + \Tr(G)=\Tr(D) = \princdiv(h_D),$$
so $h_D = N(u_i)^{\ell} h_{D_i}^{m-\ell}h_G$ up to multiplication by a constant, hence $N(u_i)^{\ell}\mid h_D$. 
If $N(u_i)$ also divides $h_{D_i}^{m-\ell}h_G$, then $\Tr(D_i) + \Tr(w(D_i)) \leq (m-\ell) \Tr(D_i) + \Tr(G)$. Since $\Tr(w(D_i)) \not\leq \Tr(G)$ is prime by Lemma~\ref{inverseimages} (ii), then $\Tr(w(D_i))=\Tr(D_i)$ and therefore $w(D_i)=\varphi^j(D_i)$ for some $j$. This yields a contradiction by (ii). Therefore, $N(u_i)^{\ell}$ exactly divides $h_D$.

Conversely, assume that $h_D = N(u_i)^{\ell} h$ for some $\ell$, where $h$ is a polynomial and $N(u_i) \nmid h$.
Then $\Tr(D) = \princdiv(h_D) = \ell \Tr(D_i) + \ell \Tr(w(D_i)) + \princdiv(h)$, and $\Tr(D_i)+\Tr(w(D_i))\not\leq \princdiv(h)$. Say e.g. that $\Tr(w(D_i)) \not\leq \princdiv(h)$, and $k$ is maximal such that $k \Tr(D_i) \leq \princdiv(h)$. Then
$$ \Tr(D) = m \Tr(D_i) + \ell \Tr(w(D_i)) + F$$
where $m = \ell+k$ and $\Tr(D_i), \Tr(w(D_i)) \not\leq \princdiv(h)-k\Tr(D_i)=:F$. By Theorem~\ref{thm:hd} (iv), $F=\Tr(D)-m\Tr(D_i)-\ell\Tr(w(D_i))=\Tr(G)$, 
where $G\in\Div^0$ is a reduced divisor with poles only at $\O$ of the form $$G=D-\sum_{l=1}^m \varphi^{a_l}(D_i)-\sum_{l=1}^{\ell}\varphi^{b_l}(D_j)$$ for some $a_l, b_l\in\{0,\ldots,n-1\}$.
\end{proof}

\begin{remark} \label{rmk:char2}
The results in this section may be generalized to elliptic and hyperelliptic curves over fields of characteristic 2 by defining $H_D = h_D (h_D \circ w)$. It is easy to check that we obtain a function $h_D$ with the same properties as in Theorem~\ref{thm:hd} and Corollary~\ref{cor:repr}. Some caution is needed in adapting Theorem~\ref{thm:hd2}. 
\end{remark}

\subsection{Computing the rational function}

It is easy to compute $h_D$ using Cantor's Algorithm (see \cite{cantor-87}) and a generalization of Miller's Algorithm (see \cite{miller-04}) as follows. For $[D_1], [D_2] \in \pic$ given in Mumford representation, Cantor's Algorithm returns a reduced divisor $D_1 \oplus D_2$ and a function $a$ such that $D_1 + D_2 = D_1 \oplus D_2 + \princdiv(a)$. We denote this as $\cantor(D_1,D_2) = (D_1 \oplus D_2,a)$. For completeness, we give Cantor's Algorithm in Algorithm \ref{algo:cantor}. Lines 1--3 are the composition of the divisors to be added, and the result of this is reduced in lines 4--8.

\begin{algorithm}
\begin{algorithmic}[1]
\caption{Cantor's Algorithm including rational function}
\label{algo:cantor}
 \Require $[u_1,v_1], [u_2,v_2] \in \pic$ in Mumford representation
 \Ensure $[u,v]$ in Mumford representation and $a$ such that $[u,v] + \princdiv(a) = [u_1,v_1] + [u_2,v_2]$
 \State $a \gets \gcd(u_1,u_2,v_1+v_2)$, find $e_1,e_2,e_3$ such that $a = e_1u_1 + e_2u_2 + e_3(v_1+v_2)$
 \State $u \gets u_1u_2/a^2$
 \State $v \gets (u_1v_2e_1 + u_2v_1e_2 + (v_1v_2+f)e_3)/a \bmod u$
 \While{$\deg u > g$}
  \State $\tilde{u} \gets \monic((f-v^2)/u), \tilde{v} \gets -v \bmod{\tilde{u}}$
  \State $a \gets a \cdot (y-v)/\tilde{u}$
  \State $u \gets \tilde{u}, v \gets \tilde{v}$
 \EndWhile
 \State \Return $[u,v], a$
\end{algorithmic}
\end{algorithm}

The following iterative definition will allow us to compute $h_D$ with a Miller-style algorithm. For a function $h$ we denote by $h^{\varphi}$ the application of the Frobenius automorphism $\varphi : \Fqbar \rightarrow \Fqbar$ coefficientwise to the function $h$. The proof of the next proposition is standard, and left to the reader.

\begin{proposition} \label{thm:recursivehd}
Let $D = [u,v]$ be a divisor on $C$, and let $D_i=\varphi^i(D)$ for $i \geq 0$.
Let $h^{(1)} = u$ as a function on $C$, and define recursively the functions
$$ h^{(i+j)} = h^{(i)} \cdot (h^{(j)})^{\varphi^i} \cdot a^{-1} $$
where $a$ is given by Cantor's Algorithm according to 
$$w(D_0 \oplus \ldots \oplus D_{i-1}) + w(D_i \oplus \ldots \oplus D_{i+j-1}) = w(D_0 \oplus \ldots \oplus D_{i+j-1}) + \princdiv(a)$$
for $i,j \geq 1$. Then for all $i \geq 1$ we have 
$$ \princdiv(h^{(i)}) = D_0 + \ldots + D_{i-1} + w(D_0 \oplus \ldots \oplus D_{i-1}). $$
If $[D] \in T_n$, then $$h^{(n-1)} = h_D.$$
\end{proposition}

Algorithm \ref{algo:millercantor} takes as an input the Mumford representation of $[D]\in T_n$ and the binary representation of $n-1$, and returns the function $h_D$.

\begin{algorithm}
\begin{algorithmic}[1]
\caption{Miller-style double and add algorithm for computing $h_D$}
\label{algo:millercantor}
 \Require $[D] = [u,v] \in T_n$ and $n-1 = \sum_{j=0}^s n_j 2^j$
 \Ensure $h_D$
 \State $h \gets u, R \gets w(D), Q \gets w(\varphi(D)), i \gets 1$
 \For {$j = s-1, s-2, \ldots,1, 0$}
   \State $(R,a) \gets \cantor(R,\varphi^i(R)), h \gets h \cdot h^{\varphi^i} \cdot a^{-1}, Q \gets \varphi^i(Q), i \gets 2i$
   \If {$n_j = 1$}
     \State $(R,a) \gets \cantor(R,Q), h \gets h \cdot u^{\varphi^i} \cdot a^{-1}, Q \gets \varphi(Q), i \gets i+1$
   \EndIf
 \EndFor
 \State \Return $h$
\end{algorithmic}
\end{algorithm}

\begin{remark}\label{largeg}
It is also possible to determine the coefficients of $h_D$ by solving a linear system of size about $gn \times gn$. 
\end{remark}

\subsection{Compression and decompression algorithms}

We propose the compression and decompression algorithms detailed in Algorithms \ref{algo:compr} and \ref{algo:decompr}. We denote by $\lc$ the leading coefficient of a polynomial. 
We only discuss the case $n\geq 3$, since in the case $n=2$ the representation consists of $u(x)$ as seen in Proposition~\ref{n=2}.

The compression algorithm follows immediately from Corollary~\ref{cor:repr} and Algorithm~\ref{algo:millercantor}.
The strategy of the decompression algorithm is as follows. From the input $\alpha = \mathcal{R}(D)$, we recompute $h_{D,1}$ and $h_{D,2}$, and then $H_D$. Then we factor $H_D$ in order to obtain the $u$-polynomials of (one Frobenius conjugate of each of) the $\Fqn$-rational prime divisors in $D$. This is consistent with the fact that $\Tr(D)$ only contains information about the conjugacy classes of these prime divisors. Afterwards, we compute the corresponding $v$-polynomial for each $u$-polynomial. In this way, if $D=D_1+\ldots+D_t$ is the decomposition of $D$ as a sum of  $\Fqn$-rational prime divisors, for each $i\in\{1,\ldots,t\}$ we recover one of the Frobenius conjugates of $D_i$, which we denote by $D'_i$. 
The divisor $D'_1+\ldots+D'_t$ corresponds to the class $\mathcal{R}^{-1}(\alpha)$ by Theorem~\ref{thm:hd} (iv). 
We always compute a reduced representative $D'_1+\ldots+D'_t$ of the class $\mathcal{R}^{-1}(\alpha)$, as discussed in Remark~\ref{rem:reduced}. 

\begin{algorithm}
\begin{algorithmic}[1]
\caption{Compression, $n\geq 3$}
\label{algo:compr}
 \Require $[D] = [u,v] \in T_n$
 \Ensure Representation $(\alpha_0,\ldots,\alpha_{(n-1)g}) \in \Fq^{(n-1)g} \times \F_2$ of $[D]$
 \State $r \gets\deg u$ 
 \State compute $h_D(x,y) = h_{D,1}(x) + y h_{D,2}(x)$ (see Algorithm \ref{algo:millercantor} and Remark~\ref{largeg})
 \State $d_1 \gets \lfloor \frac{ng}{2} \rfloor$
 \State $d_2 \gets \lfloor \frac{ng-2g-1}{2} \rfloor$
 \If{$r$ even}
 \State $h_{D,1} \gets h_{D,1}/\lc(h_{D,1})$ \Comment{Notation: $h_{D,1} =\gamma_{d_1} x^{d_1} + \gamma_{d_1-1}x^{d_1-1}+\ldots+\gamma_1x+\gamma_0 $ monic}
 \State $h_{D,2} \gets h_{D,2}/\lc(h_{D,1})$ \Comment{Notation: $h_{D,2} = \beta_{d_2}x^{d_2} + \beta_{d_2-1}x^{d_2-1} + \ldots + \beta_1x+\beta_0$}
 \Else
 \State $h_{D,1} \gets h_{D,1}/\lc(h_{D,2})$ \Comment{Notation: $h_{D,1} = \gamma_{d_1}x^{d_1} + \gamma_{d_1-1}x^{d_1-1}+\ldots+\gamma_1x+\gamma_0 $}
 \State $h_{D,2} \gets h_{D,2}/\lc(h_{D,2})$ \Comment{Notation: $h_{D,2} =\beta_{d_2} x^{d_2} + \beta_{d_2-1}x^{d_2-1} + \ldots + \beta_1x+\beta_0$ monic}
 \EndIf
  \If{$g$ even}
 \State \Return $(\beta_0,\ldots,\beta_{d_2},\gamma_0,\ldots,\gamma_{d_1})$
 \Else
 \State \Return $(\gamma_0,\ldots,\gamma_{d_1},\beta_0,\ldots,\beta_{d_2})$
 \EndIf
\end{algorithmic}
\end{algorithm}

\begin{algorithm}
\begin{algorithmic}[1]
\caption{Decompression, $n\geq 3$}
\label{algo:decompr}
 \Require $(\alpha_0,\ldots,\alpha_{(n-1)g}) \in \Fq^{(n-1)g} \times \F_2$
 \Ensure one reduced $D\in\Div^0(\Fqn)$ such that $[D] \in T_n$ has representation $(\alpha_0,\ldots,\alpha_{(n-1)g})$
 \State $d_1 \gets \lfloor \frac{ng}{2} \rfloor$
 \State $d_2 \gets \lfloor \frac{ng-2g-1}{2} \rfloor$
 \If{$g$ even}
 \State $h_{D,1}(x) \gets \alpha_{(n-1)g}x^{d_1} + \ldots + \alpha_{d_2+2}x + \alpha_{d_2+1}$
 \State $h_{D,2}(x) \gets \alpha_{d_2}x^{d_2} + \alpha_{d_2-1}x^{d_2-1} + \ldots + \alpha_1 x + \alpha_0$
 \Else
 \State $h_{D,1}(x) \gets \alpha_{d_1}x^{d_1} + \ldots + \alpha_1 x + \alpha_0 $
 \State $h_{D,2}(x) \gets \alpha_{(n-1)g}x^{d_2} + \ldots + \alpha_{d_1+2}x + \alpha_{d_1+1}$
 \EndIf
 \State $H_D(x) \gets h_{D,1}(x)^2 - f(x)h_{D,2}(x)^2$
 \State factor $H_D(x) = U_1(x)^{e_1} \cdot \ldots \cdot U_{m}(x)^{e_{m}}$ with $U_i \in \Fq[x]$ irreducible and pairwise distinct, $e_i \in \{1,\ldots,gn\}$
 \State $L \gets $ empty list
 \For{$i=1,\ldots,m$}
 \If{$U_i(x)$ is irreducible over $\Fqn$} \Comment{$U_i$ comes from an $\Fq$-rational prime divisor}
 \State $e_i \gets e_i/n$
 \EndIf 
 \State $U(x) \gets $ one irreducible factor over $\Fqn$ of $U_i(x)$ 
 \If{$h_{D,2}(x) \not\equiv 0 \bmod{U(x)}$}
 \State $V(x) \gets -h_{D,1}(x)h_{D,2}(x)^{-1} \bmod{U(x)}$
 \State append $[U(x),V(x)]$ to $L, e_i$ times
 \Else \Comment{$h_{D,2}(x) \equiv 0 \bmod{U(x)}$}
 \If{$f(x) \equiv 0 \bmod{U(x)}$} \Comment{$V(x)=0$ and $D_i=w(D_i)$}
 \State append $[U(x),0], [U(x)^{\varphi},0],\ldots,[U(x)^{\varphi^{e_i-1}},0]$ to $L$
 \Else \Comment{$V(x)\neq 0$ and $D_i\neq w(D_i)$}
 \State compute $s, h_{\Delta}$ such that $h_D = U_i(x)^s h_{\Delta}$ and $U_i(x) \nmid h_{\Delta}$
 \If{$s<e_i/2$} 
 \State $V(x) \gets -h_{\Delta,1}(x) h_{\Delta,2}(x)^{-1} \bmod U(x)$ 
 \State append $[U(x),V(x)]$ to $L, e_i-s$ times
 \State append $[U(x)^{\varphi},-V(x)^{\varphi}]$ to $L, s$ times
 \Else \Comment{$s=e_i/2$}
 \State $V(x) \gets \sqrt{f(x)} \bmod{U(x)}$
 \State append $[U(x),V(x)],[U(x)^{\varphi},-V(x)^{\varphi}]$ to $L, s$ times
 \EndIf
 \EndIf
 \EndIf
 \EndFor
 \Comment{Notation: $L = [D_1,\ldots,D_{t}]$}
 \State \Return $D=D_1+\ldots+D_t$ 
\end{algorithmic}
\end{algorithm}

It is easy to see that both algorithms terminate in polynomial time in $\log q$. Correctness of the compression algorithm follows from Proposition~\ref{thm:recursivehd}.
We now show that the decompression algorithm returns the correct output. 

\begin{theorem}
Decompression Algorithm \ref{algo:decompr} operates correctly, i.e.\ for any input $\mathcal{R}(D)$, where $[D] \in T_n$, it returns a reduced divisor $D'$ such that $[D'] \in T_n$ and $\mathcal{R}(D) = \mathcal{R}(D')$. 
\end{theorem}

\begin{proof}
Let $D=D_1+\ldots+D_t$, where $D_i$ are reduced prime divisors defined over $\Fqn$. If $D_i=\varphi^j(D_k)$ for some $k\neq i$, then $\mathcal{R}(D)=\mathcal{R}(\tilde{D})$ where $\tilde{D}=\sum_{j\neq i,k} D_j+2D_i$. $\tilde{D}$ is reduced if $D_i\neq w(D_i)$. If that is the case, we may assume without loss of generality that \begin{equation}\label{wlog}
 D_i\neq\varphi^j(D_k) \mbox{ for any } k\neq i.\end{equation}
Let $[u_i,v_i]$ be the Mumford representation of $D_i$, $u_i\in\Fqn[x]$ irreducible. We have
$$H_D(x)
=\prod_{i=1}^t u_i^{1+\varphi+\ldots+\varphi^{n-1}}=\prod_{i=1}^m U_i(x)^{e_i},$$ 
 where  $U_i\in\Fq[x]$ are irreducible and $U_i\neq U_j$ if $i\neq j$, $m\leq t$. Up to reindexing, $U_i=u_i$ if $u_i\in\Fq[x]$ and $U_i=N(u_i)$ otherwise, for $i\leq m$. If $u_i\in\Fq[x]$, then $u_i^{1+\varphi+\ldots+\varphi^{n-1}}=u_i^n=U_i^n$, hence $n\mid e_i$ and we replace $e_i$ by $e_i/n$, since $\Tr(D_i)=nD_i$. Notice that  by Lemma \ref{inverseimages} (ii) $U_i$ is an $\Fq[x]$-irreducible factor of $H_D(x)$ independently of whether $u_i\in\Fq[x]$ or not. Notice moreover that $u_i\in\Fq[x]$ if and only if $U_i$ is irreducible in $\Fqn[x]$. If $U_i$ is reducible in $\Fqn[x]$, then $u_i\in\Fqn[x]$ is one of its irreducible factors. Summarizing, each $D_i$ corresponds exactly to a set of $n$ $\Fqn[x]$-irreducible factors of $H_D$, and these factors can be correctly grouped by first computing the $\Fq[x]$-factorization of $H_D=N(u)$.
 
Fix $i\in\{1,\ldots,m\}$ and let $U(x)$ be an $\Fqn[x]$-irreducible factor of $U_i(x)$,  i.e., $U(x)$ is a Frobenius conjugate of $u_i(x)$. 
If $U\nmid h_{D,2}$ there exist polynomials $k(x),l(x)\in\Fqn[x]$ such that $k(x)h_{D,2}=1+l(x)U(x)$. Hence $k(x)(h_{D,1}(x)+yh_{D,2}(x))\equiv y+k(x)h_{D,1} \bmod U.$ Since $h_{D,1}+yh_{D,2}\equiv 0\bmod (U,y-V)$, then $V+k(x)h_{D,1}\equiv 0\bmod U$, hence $$V\equiv -h_{D,1}h_{D,2}^{-1}\bmod U.$$ 
Since $U\nmid h_{D,2}$, by Theorem~\ref{thm:hd2} (i) no Frobenius conjugate of $w(D_i)$ appears among $D_1,\ldots,D_t$. 
Notice that in particular $D_i \neq w(D_i)$, hence $V\neq 0$. Therefore, $D_i$ appears in $D$ with multiplicity $e_i$ under assumption~(\ref{wlog}).

If $U\mid h_{D,2}$, it follows from Theorem~\ref{thm:hd2} (i) that $w(D_i)=\varphi^j(D_k)$ for some $0\leq j\leq n-1$ and $1\leq k\leq t$. 
We distinguish the cases when $D_i= w(D_i)$ or $D_i \ne w(D_i)$.
The case when $D_i=w(D_i)$ is treated in lines 22--23 of the algorithm. Since $y^2-f\in (U,y-V)$, then $V^2\equiv f\bmod U$. Therefore $f\equiv 0\bmod U$ if and only if $V=0$, which is equivalent to $D_i=w(D_i)$ is equivalent to $v_i=0$. Practically, one can decide whether $D_i=w(D_i)$ by checking whether $U \mid f$. If this is the case, it suffices to set $V=0$. Since $U^{e_i}$ exactly divides $H_D$, $D_i$ and its Frobenius conjugates appear in $D$ with total multiplicity $e_i$. The divisor $D$ is reduced, therefore it must contain in its support $e_i$ distinct Frobenius conjugates of $D_i$, e.g. $D_i,\varphi(D_i),\ldots,\varphi^{e_i-1}(D_i)$, each with multiplicity one.

The last case is treated in lines 25--33 of the algorithm. In this case $D_i \ne w(D_i)$, but $w(D_i)=\varphi^j(D_k)$ for some $k\in\{1,\ldots,t\}$. This is equivalent to $U\mid h_{D,2}$ and $U\nmid f$, as we proved above. Since $n \ne 2$, then $k\ne i$ by Theorem \ref{thm:hd2} (ii). In addition, since $D$ is reduced, then $D_k\neq w(D_i)$, hence $D_i, D_k\not\in\pic(\F_q)$ and $U_i=N(U)$, $U\in\F_{q^n}[x]\setminus\F_q[x]$.
Write $\Tr(D)=m\Tr(D_i)+\ell\Tr(w(D_i))+\Tr(G)$ for some $m,\ell>0$ such that $\Tr(D_i),\Tr(w(D_i))\not\leq\Tr(G)$. By Theorem~\ref{thm:hd2} (iii), $s:=\min\{m,\ell\}$ may be computed as the exponent for which $U_i^s\mid h_D$ and $U_i^{s+1}\nmid h_D$. Equivalently, among $D_1,\ldots,D_t$ there are at least $s$ Frobenius conjugates of $D_i$ (including $D_i$) and at least $s$ Frobenius conjugates of $w(D_i)$ (including $D_k$). No divisor can be a Frobenius conjugate of both, and for one among $D_i$ and $w(D_i)$ the multiset $\mathcal{D}=\{D_1,\ldots,D_t\}$ contains exactly $s$ of its Frobenius conjugates. Remove $s$ of the Frobenius conjugates of $D_i$ and $s$ of the Frobenius conjugates of $w(D_i)$ from $\mathcal{D}$, and let $\Delta$ be the sum of the remaining divisors, counted with the multiplicity in which they appear in the multiset. Then $h_D=U_i^sh_{\Delta}$, where $h_{\Delta}=h_{\Delta,1}+yh_{\Delta,2}$ corresponds to the divisor $\Delta$. 
By Theorem~\ref{thm:hd2} (i), $U\nmid h_{\Delta,2}$, since $$\Tr(D_i)+\Tr(w(D_i))\not\leq\princdiv(h_{\Delta})=\Tr(\Delta)=(m-s)\Tr(D_i)+(\ell-s)\Tr(w(D_i))+\Tr(G).$$
If $s=e_i/2$, then the support of $D$ contains $e_i/2$ Frobenius conjugates of $D_i$ and $e_i/2$ Frobenius conjugates of $D_k$. Since it contains $e_i$ Frobenius conjugates of $D_i$ and $D_k$ in total, then $s=m=\ell$ and $V$ may be computed as $\sqrt{f} \bmod U$. Then $D$ contains exactly $e_i/2$ Frobenius conjugates of $[U,V]$ and $e_i/2$ Frobenius conjugates of $[U,-V]$. Notice that in this situation we do not need to distinguish between (Frobenius conjugates of) $D_i$ and $w(D_i)$, since they appear in $D$ with the same multiplicity.
If $s<e_i/2$, then $h_D=U_i^{\ell}h_{\Delta}$ and $\Delta$ contains $e_i-2s$ Frobenius conjugates of one among $D_i$ and $w(D_i)$. We already showed that $U\nmid h_{\Delta,2}$, hence the $V$ polynomial of the divisor which appears in $\Delta$ can be computed as $V=-h_{\Delta,1}{h_{\Delta,2}}^{-1}\bmod U$. In this case, $D$ contains $s$ Frobenius conjugates of $[U,-V]$ and $e_i-s$ Frobenius conjugates of $[U,V]$. 

Finally, we show that the divisor returned by Algorithm~\ref{algo:decompr} is reduced. To this end, we check that the algorithm does not add both a divisor and its involution to the list $L$, and in particular when a divisor is $2$-torsion, we check that it is added with multiplicity 1. Since for each $i$ such that $U\nmid h_{D,2}$ we have computed a unique $V\neq 0$, we only need to consider the cases where $U\mid h_{D,2}$. 
In the case when $U\mid f$ we have $D_i = w(D_i)$. Since $D$ is reduced, then $e_i \leq n$, and if $e_1\neq 1$ then $D_i\not\in\pic(\F_q)$.
In particular, $D_i,\varphi(D_i),\ldots,\varphi^{e_i-1}(D_i)$ are distinct.  
If $U\nmid f$, then we showed that $D_i,\varphi(D_i)\ne w(D_i)$ and $D_i\neq \varphi(D_i)$.
The divisors $D_i=[U,V]$ and $w(\varphi(D_i))=[U^{\varphi}, -V^{\varphi}]$ can be added with multiplicity greater than one since they are not 2-torsion and not one the involution of the other. 
\end{proof}

\subsection{Group operation}

An important question in the context of point compression is how to perform the group operation. For some compression methods for (hyper)elliptic curves, formulas or algorithms for performing the group operation in compressed coordinates are available. For example, the Montgomery ladder (see \cite{montgomery-87}) computes the $x$-coordinate of an elliptic curve point $kP$ from the $x$-coordinate of $P$. This method may be generalized to genus 2 hyperelliptic curves (see \cite{gaudry-07}). There is also an algorithm to compute pairings using the $x$-coordinates of the input points only (see \cite{galbraith-lin-09}). 

In such a situation, the crucial question is whether it is more efficient to perform the operation in the compressed coordinates, or to decompress, perform the operation in the full coordinates, and compress again.
Implementation practice shows that it is usually more efficient to use the second method (at least when side-channel attack resistance is not crucial), and most recent speed records for scalar multiplication on elliptic curves have been set using algorithms that need the full point, see e.g.\ \cite{ed25519,longa-sica-12,lambda,faz-hernandez-13}. Timings typically ignore the additional cost for point decompression, but there is strong evidence that on a large class of elliptic curves the second approach is faster. Moreover, Galbraith and Lin show in \cite{galbraith-lin-09} that for computing pairings, the second approach is faster whenever the embedding degree is greater than 2.

In this paper we do not provide an efficient algorithm for scalar multiplication of compressed elements of the trace zero subgroup. However, we believe that this is not a major drawback. On the basis of the results outlined above, we expect that the second method would be faster, and hence it is reasonable to use this method when computing with compressed elements of a trace zero subgroup: Decompress the element, perform the operation in $\pic(\Fqn)$, and compress the result. Since our compression and decompression algorithms are very efficient, this adds only little overhead. Moreover, scalar multiplication is considerably more efficient for trace zero divisors than for general divisors in $\pic(\Fqn)$, due to a speed-up using the Frobenius endomorphism, as pointed out by Frey \cite{frey-99} and studied in detail by Lange \cite{lange-phd,lange-04} and subsequently by Avanzi and Cesena~\cite{avanzi-cesena-07}.

\section{Representation for elliptic curves} \label{sec:ec}

Elliptic curves are simpler and better studied than hyperelliptic curves. In particular, the Picard group of an elliptic curve is isomorphic to the curve itself. Therefore one can work with the group of points of the curve, and point addition is given by simple, explicit formulas. Finding a rational function with a given principal divisor can also be made more efficient. For all these reasons, the results and methods from Section \ref{sec:rep} can be simplified and made explicit for the family $\T_{n,1}$ of trace zero varieties of elliptic curves, with respect to a field extension of fixed degree $n$.

Let $E: y^2 = f(x)$ denote an elliptic curve defined over $\Fq$. The trace zero subgroup $T_n$ of $E(\Fqn)$ is then the group of all points $P$ with trace {\it equal} to zero. We consider only $n \geq 3$, and refer to~\cite{gorla-massierer-1} for the case $n=2$.

\begin{notation}\label{not}
Write $P_i = \varphi^i(P)$ for $i = 0,\ldots,n-1$. Let $\ell_i(x,y) = 0, i = 1,\ldots,n-2,$ be the equation of the line passing through the points $P_0 \oplus \ldots \oplus P_{i-1}$ and $P_i$. Let $v_i(x,y) = 0, i = 1,\ldots,n-3,$ be the equation of the vertical line passing through the point $P_0 \oplus \ldots \oplus P_i$. 
\end{notation}

The following is obtained from Theorems~\ref{thm:hd} and~\ref{thm:hd2} in the case that the curve is elliptic. The proof that $h_P$ has the form claimed is an easy calculation, which is left to the reader.

\begin{corollary} \label{cor:hp}
Let $n \geq 3$ prime. For any $P \in T_n \setminus \{ \O \}$, let $$h_P = \frac{\ell_1 \cdot \ldots \cdot \ell_{n-2}}{v_1 \cdot \ldots \cdot v_{n-3}} \in \Fq(E),$$
where $\ell_j$ and $v_j$ are the lines defined in Notation \ref{not}. Then:
\begin{enumerate}
 \item $\princdiv(h_P) = P_0 + \ldots + P_{n-1} - n\O$.
 \item $h_P(x,y) =  h_{P,1}(x) + y h_{P,2}(x) $ for some $h_{P,1}, h_{P,2}\in \Fq[x]$.
 \item $H_P = h_{P,1}^2 - fh_{P,2}^2$ has degree $n$, and its zeros are exactly the $x$-coordinates of $P_0,\ldots,P_{n-1}$.
 \item $\deg h_{P,1} \leq \frac{n-1}{2}$ and $\deg h_{P,2} = \frac{n-3}{2}$.
 \item If $Q$ is such that $h_P = h_Q$, then $Q = \varphi^j(P)$ for some $j \in \{0,\ldots,n-1\}$.
 \item $h_{P,2}(X) \ne 0$ for all $x$-coordinates of $P_0,\ldots,P_{n-1}$.
\end{enumerate}
\end{corollary}

Since the exact degree of $h_{P,2}$ is known, $h_P$ can be normalized by making $h_{P,2}$ monic, as in Corollary~\ref{cor:repr}. One obtains the following optimal representation for trace zero points on an elliptic curve.

\begin{corollary}\label{cor:repr_ec}
Let $n \geq 3$ prime, let $d_1 = (n-1)/2, d_2 = (n-3)/2$. Write
$ h_{P,1} = \gamma_{d_1}x^{d_1} + \ldots + \gamma_0$ and $h_{P,2} = x^{d_2}+\beta_{d_2-1}x^{d_2-1}+\ldots+\beta_0$. Define
$$\begin{array}{rcl}
\mathcal{R} : T_n\setminus\{\O\} & \longrightarrow & \mathbb{F}_q^{n-1}\\
\mbox{$P$} & \longmapsto & (\gamma_0,\ldots,\gamma_{d_1},\beta_0,\ldots,\beta_{d_2-1}).
\end{array}$$
Then $$\mathcal{R}^{-1}(\mathcal{R}(P))=\{P,\varphi(P),\ldots,\varphi^{n-1}(P)\} \mbox{ for all } P\in T_n\setminus\{\O\}$$
and $\mathcal{R}$ yields an optimal representation for the family $\T_{n,1}$. \end{corollary}

One also can give simplified compression and decompression algorithms.

\begin{algorithm}
\begin{algorithmic}[1]
\caption{Compression for elliptic curves, $n\geq 3$}
\label{algo:compression}
 \Require $P \in T_n$
 \Ensure representation $(\alpha_0,\ldots,\alpha_{n-2}) \in \F_q^{n-1}$ of $P$
 \State compute $h_P(x,y) = h_{P,1}(x) + y h_{P,2}(x) \gets \frac{\ell_1 \cdot \ldots \cdot \ell_{n-2}}{v_1 \cdot \ldots \cdot v_{n-3}}(x,y)$ (see Algorithm \ref{algo:miller}) where
 \State $h_{P,1}(x) = \gamma_{d_1}x^{d_1} + \ldots +\gamma_0$ and 
 \State $h_{P,2}(x) = x^{d_2}+\beta_{d_2-1}x^{d_2-1}+\ldots+\beta_0$
 \State \Return $(\gamma_0,\ldots,\gamma_{d_1},\beta_0,\ldots,\beta_{d_2-1})$
\end{algorithmic}
\end{algorithm}

\begin{algorithm}
\begin{algorithmic}[1]
\caption{Decompression for elliptic curves, $n \geq 3$}
\label{algo:decompression}
 \Require $(\alpha_0,\ldots,\alpha_{n-2}) \in \F_q^{n-1}$
 \Ensure one point $P \in T_n \setminus \{\O\}$ with representation $(\alpha_0,\ldots,\alpha_{n-2})$
 \State $h_{P,1}(x) \gets \alpha_{(n-1)/2} x^{(n-1)/2} + \alpha_{(n-3)/2} x^{(n-3)/2} + \ldots + \alpha_{1} x + \alpha_{0}$
 \State $h_{P,2}(x) \gets x^{(n-3)/2} + \alpha_{n-2} x^{(n-5)/2} + \ldots + \alpha_{(n+3)/2} x + \alpha_{(n+1)/2}$
 \State $H_P(x) \gets h_{P,1}(x)^2 - f(x) h_{P,2}(x)^2$
 \State $X \gets$ one root of $H_P(x)$
 \State $Y \gets -h_{P,1}(X)/h_{P,2}(X)$
 \State \Return $P=(X,Y)$
\end{algorithmic}
\end{algorithm}

Finally, we discuss how to compute $h_P$ for different values of $n$. Explicit formulas can be computed in the special cases $n=3,5$. We do this in Appendix~\ref{app:explicitequations}. For general $n$, a straightforward computation of $h_P$ is possible, since Corollary~\ref{cor:hp} contains an explicit formula given in terms of lines. Such a computation can be made more efficient by employing the usual divide and conquer strategy. Computing $h_P$ via a Miller-style algorithm analogous to Algorithm \ref{algo:millercantor} is also possible. The latter is advantageous for medium and large values of $n$, while for small values of $n$ a straightforward computation using a divide and conquer approach seems preferable (unless explicit formulas are available). According to our experiments, a Miller-style algorithm behaves better than the obvious way of computing $h_P$ (i.e.\ iteratively multiplying by $\frac{\ell_i}{v_{i-1}}$) for $n > 10$, and better than a divide and conquer approach for $n > 20$.

We denote by $\ell_{P,Q}$ the line through the points $P$ and $Q$, and by $v_P$ the vertical line through $P$. All computations are done with functions on $E$, i.e.\ in $\Fqn(E)$.

\begin{algorithm}
\begin{algorithmic}[1]
\caption{Miller-style double and add algorithm for computing $h_P, n \geq 3$}
\label{algo:miller}
 \Require $P \in T_n \setminus \{\O\}$ and $n-1 = \sum_{j=0}^s n_j 2^j$
 \Ensure $h_P$
 \State $Q \gets \varphi(P)$
 \State $h \gets \ell_{P,Q},~ R \gets P\oplus Q,~ Q \gets \varphi(Q),~ i \gets 2$
 \If {$n_{s-1} = 1$}
    \State $h \gets h \cdot \frac{\ell_{R,Q}}{v_R},~ R \gets R\oplus Q,~ Q \gets \varphi(Q),~ i \gets 3$
 \EndIf
 \For {$j = s-2, s-3, \ldots,1, 0$}
   \State $h \gets h \cdot h^{\varphi^i} \cdot \frac{v_{R + \varphi^i(R)}}{\ell_{w(R),w(\varphi^i(R))}},~ R \gets R \oplus \varphi^i(R),~ Q \gets \varphi^i(Q),~ i \gets 2i$
   \If {$n_j = 1$}
     \State $h \gets h \cdot \frac{\ell_{R,Q}}{v_R},~ R \gets R\oplus Q,~ Q \gets \varphi(Q),~ i \gets i+1$
   \EndIf
 \EndFor
 \State \Return $h$
\end{algorithmic}
\end{algorithm}

\section{Timings and comparison with other representations}\label{sec:timings}

This new representation applies to any prime $n$ and any genus, and it can be made practical for very large values of $n$ and/or $g$.
Moreover our decompression algorithm allows the unique recovery of one well-defined class of conjugates of the original point. For elliptic curves, such a class consists exactly of the Frobenius conjugates of the original point, and for higher genus curves, classes are as described in Theorem~\ref{thm:hd} (iv). Identifying these conjugates is the natural choice from a mathematical point of view, since it respects the structure of our object and is compatible with scalar multiplication. 

There are only three other known methods for point compression in trace zero varieties over elliptic curves, namely \cite{naumann}, \cite{silverberg-05}, and \cite{gorla-massierer-1}. While \cite{naumann} only applies to extension degree 3, \cite{silverberg-05, gorla-massierer-1} can be made practical for $n=3,5$. The approach of \cite{gorla-massierer-1} allows unique recovery of an equivalence class for $n=3$ and for most points for $n=5$. The compression method of \cite{silverberg-05} identifies sets of points which are incompatible with scalar multiplication, thus requiring extra bits to resolve ambiguity. There is only one known method for point compression in trace zero varieties over hyperelliptic curves from \cite{lange-04}. This method can be made practical for the parameters $g=2, n=3$. 

One advantage of our representation with respect to the previous ones is that it is the only one that does not identify the positive and negative of a point, thus allowing a recovery of the $y$-coordinate of a compressed point that does not require computing square roots. For small values of $n$, this gives a noticeable advantage in efficiency. In addition, our method works for all affine points on  the trace zero variety, without having to disregard a closed subset as is done in \cite{silverberg-05, lange-04}. In addition, our compression and decompression algorithms do not require a costly precomputation, such as that of the Semaev polynomial in \cite{gorla-massierer-1} or the elimination of variables from a polynomial system in \cite{lange-04}. 

In terms of efficiency, our compression algorithm is slower than all the other ones for elliptic curves, but our decompression algorithm is faster in all cases. For $g=1$, the time for compression and decompression together is comparable for $n=3$, and smaller for $n=5$, than that of \cite{gorla-massierer-1}. That is to say, the faster decompression makes up for the slower compression. Although in this paper we concentrate on the case of odd characteristic, our method can be adapted to fields of even characteristic, just like all other methods from \cite{gorla-massierer-1,silverberg-05,lange-04,naumann}.

We now compare the efficiency of our algorithms with those of \cite{gorla-massierer-1,silverberg-05,lange-04,naumann} in more detail. The comparison of our method with that of \cite{gorla-massierer-1} is on the basis of a precise operation count, complexity analysis, and our own Magma implementations. Notice that our programs are straightforward implementations of the methods described here and in \cite{gorla-massierer-1}, and they are only meant as an indication. No particular effort has been put into optimizing them, and clearly a special purpose implementation (e.g.\ choosing $q$ of a special shape) would produce better and more meaningful results. All computations were done with Magma version 2.19.3 \cite{magma}, running on one core of an Intel Xeon Processor X7550 (2.00 GHz). Our timings are average values for one execution of the algorithm, where averages are computed over $10 000$ executions with random inputs. Our comparison with \cite{naumann,silverberg-05,lange-04} is rougher, since no precise operation counts, complexity analyses or implementations of those methods are available. 

\medskip \noindent {\bf Comparison and Timings for $g=1, n=3$.}
We compare our method with the most efficient method from \cite{gorla-massierer-1} (there called ``compression in $t_i$'') in terms of operations in Table \ref{operations3} and timings in Table \ref{timings3}. We choose arbitrary elliptic curves such that the associated trace zero subgroups have prime order for fields of 20, 40, 60, and 79 bits. We see that the compression algorithm from~\cite{gorla-massierer-1} requires fewer operations, but we could not observe a significant difference in the timings (probably due to insufficient accuracy of our tests). 
For the decompression algorithm, we compare ``full decompression'', where one entire point (including the $y$-coordinate) is recomputed. Here, the method of \cite{gorla-massierer-1} is much slower (roughly by a factor 10), due to the necessary square root extraction. This shows one major efficiency advantage of the approach that we follow in this paper: Recovering the $y$-coordinate is much faster, since no square root computation is necessary. For a different point of view, we also compare ``decompression in $x$ only'', where no $y$-coordinate is computed. In this case, the algorithm proposed in this paper and the one from \cite{gorla-massierer-1} behave similarly.

In \cite{silverberg-05}, compression is free. The bulk of the work in the decompression algorithm is factoring a degree 4 polynomial and recomputing the $y$-coordinate from the curve equation (which requires a square root extraction). This is clearly more expensive than the decompression algorithm in this paper, which does not require polynomial factorization or square root extraction. We refer to~\cite[Section 5]{gorla-massierer-1} for a detailed discussion of the decompression algorithm from~\cite{silverberg-05}.

Naumann \cite{naumann} does not give explicit compression or decompression algorithms, but he derives an equation for the trace zero subgroup that might be used for such. The equation is in the Weil restriction coordinates $x_0,x_1,x_2$ of the $x$-coordinate of a trace zero point, and it has degree 4 in $x_0$ and degree 3 in $x_1, x_2$. Therefore, it allows a representation in the coordinates $(x_0,x_1)$ or $(x_0,x_2)$, where decompression could be done by factoring a cubic polynomial in the missing coordinate, and then recomputing the $y$-coordinate as a square root. Again, this is clearly more expensive than the decompression algorithm in this paper.

\begin{table}
\caption{Number of operations in $\Fq$ for compression/decompression of one point when $g=1,n=3$}
\label{operations3}
\begin{tabular}{l|l}
\hline\noalign{\smallskip}
Compression                                   & 2S+6M+1I\\
Compression \cite{gorla-massierer-1}            & 1M\\
Full decompression                            & 5S+5M+1I, 1 square root, 2 cube roots\\
Full decompression \cite{gorla-massierer-1}     & 4S+3M+2I, 1 square root, 2 cube roots, and 1 square root in $\F_{q^3}$\\
Decompression $x$ only                            & 5S+4M+1I, 1 square root, 2 cube roots\\
Decompression $x$ only \cite{gorla-massierer-1}     & 4S+3M+2I, 1 square root, 2 cube roots\\
\noalign{\smallskip}\hline
\end{tabular}
\end{table}

\begin{table}
\caption{Average time in milliseconds for compression/decompression of one point when $g=1,n=3$}
\label{timings3}
\begin{tabular}{l|llll}
\hline\noalign{\smallskip}
$q$ &  $2^{20}-3$ & $2^{40}-87$ & $2^{60}-93$ & $2^{79}-67$  \\
\noalign{\smallskip}\hline\noalign{\smallskip}
Compression                                    & 0.01 & 0.03 & 0.03 & 0.04\\
Compression \cite{gorla-massierer-1}             & 0.01 & 0.02 & 0.03 & 0.04\\
Full decompression                             & 0.18 & 0.71 & 0.89 & 1.52\\
Full decompression \cite{gorla-massierer-1}      & 0.84 & 7.62 & 10.62 & 17.58\\
Decompression $x$ only                         & 0.15 & 0.63 & 0.87 & 1.40\\
Decompression $x$ only \cite{gorla-massierer-1}  & 0.15 & 0.68 & 0.87 & 1.44\\
\noalign{\smallskip}\hline
\end{tabular}
\end{table}

\medskip \noindent {\bf Comparison and Timings for $g=1, n=5$.}
A similar comparison for extension degree 5 (see Tables \ref{operations5} and \ref{timings5}) shows that the compression algorithm proposed in this paper is less efficient than that of \cite{gorla-massierer-1}, but the decompression algorithm is faster. Although the bulk of the work in both decompression algorithms is polynomial factorization, following the approach proposed in this paper we have to factor one polynomial of degree 5 over $\F_{q^5}$, where the algorithm of \cite{gorla-massierer-1} first factors a polynomial of degree 6 over $\Fq$, and then at least one polynomial of degree 5 over $\F_{q^5}$. For this reason, the decompression algorithm proposed in this paper performs better than that of \cite{gorla-massierer-1}, regardless of whether we include the recovery of the $y$-coordinate. Notice that we again compare with the best method from \cite{gorla-massierer-1}, there called ``compression/decompression in the $s_i$ with polynomial factorization''.

In comparison to \cite{silverberg-05}, our compression algorithm is less efficient, but our decompression method is more efficient. The decompression algorithm of Silverberg involves resultant computations and the factorization of a degree 27 polynomial. If one wishes to recover the $y$-coordinate, a square root extraction is also required.  With or withour square root extraction, this is much more expensive than the decompression algorithm in this paper, which does not require polynomial factorization or resultant computations.
We refer to~\cite[Section 6]{gorla-massierer-1} for a detailed analysis of the algorithm from~\cite{silverberg-05}.

\begin{table}
\caption{Number of operations/complexity for compression/decompression of one point when $g=1,n=5$}
\label{operations5}
\begin{tabular}{l|l}
\hline\noalign{\smallskip}
Compression                                   & 3S+18M+3I in $\F_{q^5}$\\
Compression \cite{gorla-massierer-1}            & 5S+13M in $\Fq$\\
Full decompression                            & $O(\log q)$ operations in $\Fq$\\
Full decompression \cite{gorla-massierer-1}     & $O(\log q)$ operations in $\Fq$, and 1 square root in $\F_{q^5}$\\
Decompression $x$ only                            & $O(\log q)$ operations in $\Fq$\\
Decompression $x$ only \cite{gorla-massierer-1}     & $O(\log q)$ operations in $\Fq$\\
\noalign{\smallskip}\hline
\end{tabular}
\end{table}

\begin{table}
\caption{Average time in milliseconds for compression/decompression of one point when $g=1,n=5$}
\label{timings5}
\begin{tabular}{l|llll}
\hline\noalign{\smallskip}
$q$ & $2^{10}-3$ & $2^{20}-5$ & $2^{30}-173$ & $2^{40}-195$  \\
\noalign{\smallskip}\hline\noalign{\smallskip}
Compression                                   & 0.21 & 0.25 & 0.46 & 0.80 \\
Compression \cite{gorla-massierer-1}            & 0.04 & 0.04 & 0.05 & 0.10 \\
Full decompression                            & 0.82 & 9.39 & 4.26 & 10.13 \\
Full decompression \cite{gorla-massierer-1}     & 5.89 & 17.90 & 30.21 & 63.60 \\
Decompression $x$ only                        & 0.77 & 9.36 & 4.01 & 9.82 \\
Decompression $x$ only \cite{gorla-massierer-1} & 5.53 & 16.48 & 21.42 & 45.08 \\
\noalign{\smallskip}\hline
\end{tabular}
\end{table}

\medskip \noindent {\bf Timings for $g=1,n>5$.} We study the performance of our algorithms by means of experimental results for $n > 5$. First, for comparison with the last column of Tables \ref{timings3} and \ref{timings5}, we give in Table \ref{timingsn160} timings for $n = 7, 11, 13, 19, 23$ and corresponding randomly chosen values of $q, A$, and $B$ that produce prime order trace zero subgroups of approximately 160 bits. From the different values for decompression times (due to the fact that the performance of the polynomial factorization algorithm in Magma depends heavily on the specific choice of $q$ and $n$), we see that there is much room for optimization in the choice of these parameters.

In each case, we choose the fastest method of computing $h_P$ during compression. According to our experiments, this is an iterative approach for $n=7$, a divide and conquer approach for $n=11,13,19$, and Algorithm \ref{algo:miller} for $n \geq 23$.  During decompression we compute the $y$-coordinate of the point as well, since the difference with computing the $x$-coordinate only is negligible.

We also report that we are able to apply our method to much larger trace zero subgroups and much larger values of $n$. More specifically, our implementation was tested on trace zero subgroups of more than 3000 bits and for values of $n$ larger than 300. For even larger values of $n$, the limitation is not our compression/decompression approach, but rather the fact that the trace zero subgroup becomes very large, even for small fields.

\begin{table}
\caption{Average time in milliseconds for compression/decompression of one point when $g=1,n>5$, $\log_2 |T_n| \approx 160$}
\label{timingsn160}
\begin{tabular}{l|lllll}
\hline\noalign{\smallskip}
$n$ & $7$ & $11$ & $13$ & $19$ & $23$  \\
$q$ & $2^{27}-27689095$ & $2^{16}-129$ & $2^{14}-6113$ & $2^{9}-55$ & $2^{8}-117$  \\
\noalign{\smallskip}\hline\noalign{\smallskip}
Compression                                   & 1.80 & 2.84 & 3.89 & 8.82 & 12.90\\
Full decompression                        & 20.90 & 10.16 & 4.03 & 119.75 & 58.15\\
\noalign{\smallskip}\hline
\end{tabular}
\end{table}

\medskip \noindent {\bf Comparison and Timings for $g=2,n=3$.} We present timings for trace zero subgroups of 20, 30, 40, 50, 60 bits in Table \ref{timings23}. 
The reason for testing only small groups is that it is difficult to produce larger ones in Magma without writing dedicated code. 
Since our implementation serves mostly as a proof of concept and for comparison purposes, we did not put much effort into producing suitable curves for larger trace zero subgroups. 

The representation of \cite{lange-04} consists of 4 (out of 6) Weil restriction coordinates of the coefficients of the $u$-polynomial of a point, plus two small numbers to resolve ambiguity. Following the notation of the original paper, we call the transmitted coordinates $u_{12},u_{11},u_{10},u_{02}$, the two small numbers $a,b$, and the dropped coordinates $u_{01},u_{00}$.
This approach requires as a precomputation the elimination of 4 variables from a system of 6 equations of degree 3 in 10 variables. The result is a triangular system of 2 equations in 6 indeterminates. The compression algorithm substitutes the values of $u_{12},u_{11},u_{10},u_{02}$ into the system and solves for the two missing values in order to determine $a,b$, which in turn determine the roots coinciding with $u_{01},u_{00}$. The decompression algorithm uses $a,b$ to decide which among the solutions of the system are the coordinates it recovers. The advantage of this algorithm is that it works entirely over $\Fq$. Nevertheless, compression is clearly less efficient than our compression algorithm, since we only need to evaluate a number of expressions, while Lange has to solve a triangular system, which involves computing roots. While our decompression algorithm requires the factorization of one or two polynomials, which has complexity $O(\log q)$, Lange's decompression algorithm solves again the same triangular system. Since this involves computing roots in $\Fq$, which has complexity $O(\log^4 q)$ using standard methods (and can be as low as $O(\log^2 q)$ for special choices of parameters, see \cite{barreto-voloch-04}), it is less efficient than the decompression algorithm proposed in this paper. Notice also that Lange's approach does not give the $v$-polynomial, which needs to be computed separately, adding to the complexity of decompression.

\begin{table}
\caption{Average time in milliseconds for compression/decompression of one point when $g=2,n=3$}
\label{timings23}
\begin{tabular}{l|lllll}
\hline\noalign{\smallskip}
$q$ & $2^5-1$ & $2^8-75$ & $2^{10}-3$ & $2^{13}-2401$ & $2^{15}-19$ \\
\noalign{\smallskip}\hline\noalign{\smallskip}
Compression                                   & 0.10 & 0.11 & 0.19 & 0.19 & 0.17 \\
Full decompression                            & 0.28 & 4.78 & 19.87 & 3.07 & 3.82 \\
\noalign{\smallskip}\hline
\end{tabular}
\end{table}

\medskip \noindent {\bf Timings for $g>2,n>3$.} As a proof of concept, we provide timings in Table \ref{timings160} for trace zero subgroups of approximately 160 bits when $n=5$ and $g = 5,6,\ldots,11$. The reason for this choice is simply that we are able to find suitable curves for these parameters. We stress again that the limitation here is not our compression method, but finding trace zero subgroups of known group order, so we expect that our method will work for much larger values of $n$ and $g$ (e.g.\ we are able to compute an example for $g=2,n=23$, where the group has 173 bits).

\begin{table}
\caption{Average time in milliseconds for compression/decompression of one point when $n=5,g \geq 5$, $\log_2 |T_n| \approx 160$}
\label{timings160}
\begin{tabular}{l|lllllll}
\hline\noalign{\smallskip}
$g$ & $5$ & $6$ & $7$ & $8$ & $9$ & $10$ & $11$ \\
$q$ & $2^{8}-5$ & $2^{7}-27$ & $2^{6}-23$ & $2^{5}-1$ & $2^{4}-5$ & $2^{4}-5$ & $2^{4}-5$ \\
\noalign{\smallskip}\hline\noalign{\smallskip}
Compression                                   & 6.53 & 7.48 & 9.89 & 11.83 & 1.90 & 2.93 & 3.24\\
Full decompression                            & 4.35 & 13.91 & 12.61 & 10.27 & 29.30 & 33.83 & 42.97\\
\noalign{\smallskip}\hline
\end{tabular}
\end{table}

\section{Conclusion} \label{sec:conclusion}

In this paper, we propose a representation of elements of the trace zero subgroup via rational functions. To the extent of our knowledge, this representation is the only one that applies to elliptic and hyperelliptic curves of any genus and field extensions of any prime degree. Our representation has convenient mathematical properties: It identifies well-defined classes of points, it is compatible with scalar multiplication, and it does not discard the $v$-polynomial of the Mumford representation (or the $y$-coordinate of an elliptic curve point), thus saving expensive square root computations in the decompression process.

Our compression and decompression algorithms are efficient, even for medium to large values of $n$ and $g$. For those parameters where other compression methods are available (namely, for very small $n$ and $g$), our algorithms are comparable with or more efficient than the previously known ones, if compression and decompression are considered together. No costly precomputation is required during the setup of the system.


{
\bibliographystyle{amsalpha}  
\bibliography{lit}  
}

\appendix

\section{Explicit equations} \label{app:explicitequations}

We compute explicit equations for compression and decompression for the cases when $g=1$ and $n=3,5$, or $g=2$ and $n=3$. We give explicit formulas for compression, while for decompression we explicitly compute a low degree polynomial, whose roots give the result of the decompression. 

In addition to making the computation more efficient, the results contained in this appendix allow us to perform precise operation counts, and thus to compare our method to the other existing compression methods in Section \ref{sec:timings}. When computing complexities, we count squarings (S), multiplications (M), and inversions (I) in $\Fq$, but not additions or multiplications by constants.

\subsection{Explicit equations for $g=1,n=3$}

In this case $h_P = \ell_1$ is a line through the points $P, \varphi(P), \varphi^2(P)$. 
We assume that $\Fq$ does not have characteristic 2 or 3 and that $E$ is given by an equation in short Weierstrass form
$$ E : y^2 = x^3 + Ax + B. $$
For simplicity, we also assume that $3 \mid q-1$ and write $\F_{q^3} = \Fq[\zeta]/(\zeta^3-\mu)$ as a Kummer extension, where $\mu \in \Fq$ is not a third power. Then $1, \zeta, \zeta^2$ is a basis of $\F_{q^3}|\Fq$. It is highly likely that there exists a suitable $\mu$ of small size, see \cite[Section 3.1]{lange-04}.
When working with a field extension where $3 \nmid q-1$, one may use a normal basis, which yields similar but denser equations.


\medskip \noindent {\bf Compression.} 
If $P = (X,Y) \notin E(\Fq)$, then the equation of $h_P = \ell_1$ is
$$h_P = y +  \gamma_1 x + \gamma_0$$ 
and $\mathcal{R}(P) = (\gamma_0, \gamma_1)\in\mathbb{F}_q^2.$
Let \begin{equation}\label{xywr}
\begin{aligned}
 X & =  X_0 + X_1 \zeta + X_2 \zeta^2  \\
 Y & =  Y_0 + Y_1 \zeta + Y_2 \zeta^2 
\end{aligned}
\end{equation}
then a simple computation yields
\begin{eqnarray*}
  \gamma_1 & = & \frac{c_1 X_1^2 Y_1 + c_2 X_2^2 Y_2}{c_1 X_1^3 + c_2 X_2^3}\\
  \gamma_0 & = & -\gamma_1 X_0 - Y_0,
\end{eqnarray*}
where 
\begin{eqnarray*}
  c_1 & = & 1 - \mu^{(q-1)/3}\\
  c_2 & = & \mu^{1+(q-1)/3} - \mu = - \mu c_1
\end{eqnarray*}are constants and can be precomputed during the setup phase of the algorithm.
Hence compression takes 2S+6M+1I in $\Fq$.

When $P \in E(\Fq)$, the line $\ell_1$ is a tangent and we have
\begin{align*}
 \gamma_1 &= \frac{3X^2+A}{2Y}\\
 \gamma_0 &= -\gamma_1X-Y.
\end{align*}
Notice that such points are in $E[3](\Fq)$ and therefore very few. 

\medskip \noindent {\bf Decompression.} 
This algorithm computes the polynomial $H_P$ and its roots over $\F_{q^3}$. We have
$$ H_P(x) = x^3 - \gamma_1^2x^2 + (A-2\gamma_0\gamma_1)x - \gamma_0^2 + B.$$
Computing the coefficients of $H_P$ therefore takes 2S+1M in $\Fq$.
Since the roots of this polynomial are $X,X^q,X^{q^2}$, and using (\ref{xywr}), we get
\begin{equation}\label{sys}
 \begin{array}{rclcl}
\gamma_1^2 & = & X + X^q + X^{q^2} & = & 3X_0 \\
A-2\gamma_0\gamma_1 & = & X^{1+q} + X^{1+q^2} + X^{q+q^2} & = & 3X_0^2 - 3\mu X_1X_2\\
\gamma_0^2 - B & = & X^{1+q+q^2} & = & X_0^3 - 3 \mu X_0 X_1 X_2 + \mu X_1^3 + \mu^2 X_2^3. 
\end{array} 
\end{equation}
Hence one can solve system (\ref{sys}) over $\Fq$, to recover $(X_0,X_1,X_2)$. Since the solutions of the system are exactly the Frobenius conjugates of $X$, it suffices to find a single solution. This takes at most 3S+3M+1I, one square root, and two cube roots in $\Fq$ (see \cite[Section 5]{gorla-massierer-1}). Notice that, since this system is so simple, this is more efficient than factoring $H_P$ over $\F_{q^3}$.
Finally, $Y = -\gamma_1 X - \gamma_0$, so recomputing one $y$-coordinate takes 1M in $\Fq$, and the other ones can be recovered via the Frobenius map.
In total, decompression takes at most 5S+5M+1I, one square root, and two cube roots in $\Fq$.

\subsection{Explicit equations for $g=1, n=5$ }

We assume that $E$ is given in short Weierstrass form $E: y^2  = x^3+Ax+B$ over a field of characteristic not equal to 2 or 3.

\medskip \noindent {\bf Compression.} 
Let $P = (X,Y) \in T_5$ and denote by $\lambda_1, \lambda_2, \lambda_3$ the slopes of the lines $\ell_1, \ell_2, \ell_3$, respectively. We have
$$ h_P = \frac{\ell_1 \ell_2 \ell_3}{v_1 v_2} = (\gamma_2 x^2 + \gamma_1 x + \gamma_0) + y (x + \beta_0), $$
where
\begin{eqnarray*}
  \gamma_2 & = & -\lambda_1 - \lambda_2 - \lambda_3\\
  \beta_0 & = & -\lambda_2 \gamma_2 + \lambda_1 \lambda_3 - X^{q^2}\\
  \gamma_1 & = & -\lambda_2 \beta_0 - \gamma_2 X^{q^2} + \lambda_1 X + \lambda_3 X^{q^3} - Y - Y^{q^2} - Y^{q^3}\\
  \gamma_0 & = & \gamma_1 (\lambda_2^2 - X^{q^2}) + \gamma_2((X+X^q)(X+X^q-X^{q^2}-2\lambda_1^2+\lambda_2^2) + \lambda_1^4 + A + \lambda_1^2 X^{q^2})\\
  && + \lambda_1 \lambda_2 \lambda_3 (X + X^{q^2} + X^{q^3}) - \lambda_1 \lambda_2 Y^{q^3} - \lambda_1 \lambda_3 Y^{q^2} - \lambda_2 \lambda_3 Y + \lambda_3 \lambda_1^2 \lambda_2^2 + \lambda_1^3 \lambda_2^2 + \lambda_1^2 \lambda_2^3.
\end{eqnarray*}
Computing $\lambda_1, \lambda_2, \lambda_3$ takes a total of 3M+3I in $\F_{q^5}$. Then, $\beta_0,\gamma_0,\gamma_1,\gamma_2$ can be computed with a total of 3S+15M in $\F_{q^5}$. Thus, compression takes a total of 3S+18M+3I in $\F_{q^5}$.

\medskip \noindent {\bf Decompression.} 
We compute
\begin{eqnarray*}
  S_1 & = & \gamma_2^2 - 2 \beta_0\\
  S_2 & = & \beta_0^2+A - 2 \gamma_1 \gamma_2\\
  S_3 & = & \gamma_1^2 +2 \gamma_0 \gamma_2 - 2A \beta_0 - B\\
  S_4 & = & A \beta_0^2 + 2B \beta_0 - 2 \gamma_0 \gamma_1\\
  S_5 & = & \gamma_0^2 - B\beta_0^2
\end{eqnarray*}
using 4S+3M in $\Fq$. Then we factor the polynomial $H_P(x) = x^5 - S_1x^4 + S_2 x^3 - S_3 x^2 + S_4 x - S_5$, which takes $O(\log_2q)$ operations in $\Fq$. Finally, recovering $Y$ costs 1S+3M+1I in $\F_{q^5}$.

\subsection{Explicit equations for $g=2, n=3$}
We assume $2, 3 \nmid |\pic(\F_{q^3})|$ and that the characteristic of $\Fq$ is not equal to 2 or 5. A simple transformation yields a curve equation of the shape
$$ C: y^2 = x^5 + f_3x^3 + f_2x^2 + f_1 x + f_0. $$
We assume that $C$ is given in this form, which slightly simplifies the equations. Formulas for the general case can be worked out similarly.

The trace zero variety of hyperelliptic curves of genus 2, with respect to a degree 3 base field extension, was studied in detail by Lange \cite{lange-phd,lange-04}. One of her results is that the Mumford representation of all non-trivial elements of $T_3$ has a $u$-polynomial of degree 2.

\begin{theorem}[{\cite[Theorem 2.2]{lange-04}}] \label{thm:shape}
Assume that $C$ has genus $2$ and that $2, 3 \nmid |\pic(\F_{q^3})|$. Then all non-trivial elements of $T_3$ are represented by reduced divisors of the form
 $$ P_1 + P_2 - 2\O \notin \Div(\F_{q}), $$
 where $P_1, P_2 \ne \O$ and $P_1 \ne P_2, \varphi(P_2), \varphi^2(P_2)$.
\end{theorem}

\begin{corollary}\label{cor:shape}
Assume that $C$ has genus $2$ and that $2,3 \nmid |\pic(\F_{q^3})|$. Then all non-trivial elements of $T_3$ are represented by reduced divisors of the form $D=P_1+P_2-2\O \notin \Div(\F_{q}),$
and one of the following mutually exclusive facts holds:
\begin{enumerate}
\item $P_1,P_2\in C(\F_{q^3})\setminus \{\O\}$ and $P_1\in\{ w(\varphi(P_2)), w(\varphi^2(P_2))\}$,
\item $P_1,P_2\in C(\F_{q^3})\setminus \{\O\}$ and $P_1 \ne P_2, \varphi(P_2), \varphi^2(P_2), w(\varphi(P_2)), w(\varphi^2(P_2))$,
\item $P_1\in C(\F_{q^6})\setminus C(\F_{q^3})$ and $P_2=\varphi^3(P_1)$.
\end{enumerate}
Let $[u,v]$ be the Mumford representation of $[D]$. Then in cases (ii) and (iii) the divisor $D+\varphi(D)$ is semi-reduced and $u\nmid h_{D,2}$, in particular $h_{D,2}\neq 0$. 
\end{corollary}

\begin{proof}
It is easy to check that (i)-(iii) are mutually exclusive, and that one must be in one of these situations. We now show that $D+\varphi(D)$ is semi-reduced and $u\nmid h_{D,2}$.
If we are in case (ii), then clearly $D+\varphi(D)$ is semi-reduced. By contradiction assume that $h_{D,2}\equiv 0 \bmod{u}$. Let $P_j=(X_j,Y_j)$, $j=1,2$. $P_j-\O\in\Div(\F_{q^3})$ is a reduced prime divisor. Since $h_{D,2}(X_j)=0$, by Theorem~\ref{thm:hd2} (i) we have $w(P_j)=\varphi^i(P_j)$. Then $X_j\in\F_{q^3}\cap\F_{q^i}=\Fq$ and $Y_j\in\F_{q^3}\cap\F_{q^{2i}}=\Fq$. Hence $D=P_1+P_2-2\O\in\Div(\Fq)$, which contradicts Theorem~\ref{thm:shape}.

Assume now that we are in case (iii). 
Since $D$ is prime, by Theorem~\ref{thm:hd2} (i), $u\mid h_{D,2}$ if and only if $w(D)=\varphi^i(D)$ for some $i=1,2$. By contradiction, assume this is the case. Then either $w(P_1)=\varphi^i(P_1)$ or $w(P_1)=\varphi^{i+3}(P_1)$. Hence $X=X^{q^j}\in\F_{q^6}\cap\F_{q^j}\subseteq\F_{q^2}$ and $Y=-Y^{q^j}\in\F_{q^6}\cap\F_{q^{2j}}\subseteq\F_{q^2}$ for some $j\in\{i,i+3\}$. This shows that $D\in\Div(\F_{q^2})\cap\Div(\F_{q^3})=\Div(\Fq)$, which contradicts Theorem~\ref{thm:shape}. Therefore $u\nmid h_{D,2}$ and $D+\varphi(D)=P_1+\varphi(P_1)+\varphi^3(P_1)+\varphi^4(P_1)-4\O$ is semi-reduced. Notice that $P_1\neq w(\varphi(P_2))$ and $P_2\neq w(\varphi(P_1))$, since $D$ is reduced.
\end{proof}

\medskip \noindent {\bf Compression.} We consider elements $0\neq [D] = [u,v] \in T_3$, $D = P_1 + P_2 - 2\O$ with $P_1\neq w(\varphi(P_2))$, $w(\varphi^2(P_2))$ and $u,u^{\varphi}$ coprime. 
The special cases can be worked out separately, and we do not treat them here.

\begin{proposition}\label{prop:v}
Let $0 \neq [D]=[u,v] \in T_3$,  $D = P_1 + P_2 - 2\O$ with $P_1\neq w(\varphi(P_2))$, $w(\varphi^2(P_2))$ and $\gcd(u,u^{\varphi})=1$. Let $[U,V]$ be the Mumford representation of the semi-reduced divisor $D + \varphi(D)$. Then $$h_D = y - V \mbox{ where } V = su+v,\; s \equiv (v^{\varphi}-v)/u \bmod {u^{\varphi}}.$$ 
\end{proposition}

\begin{proof}
The divisor $D + \varphi(D)$ is semi-reduced by Corollary~\ref{cor:shape}.
By Theorem \ref{thm:hd} (iii), we have $h_D = h_{D,1} + yh_{D,2}$ with $\deg h_{D,1} = 3$ and $\deg h_{D,2} \leq 0$. Since $h_{D,2} \ne 0$ by Corollary~\ref{cor:shape}, after multiplication by a constant we have $h_D = y - \gamma(x)$ where $\gamma \in \Fq[x]$ of degree 3. If $P_i = (X_i,Y_i)$, then $h_D(X_i^{q^j},Y_i^{q^j})=0$ and hence $\gamma(X_i^{q^j}) = Y_i^{q^j}$ for $i=1,2$, $j=0,1,2$. But $V$ is the unique polynomial of degree $\leq 3$ with $V(X_i^{q^j}) = Y_i^{q^j}$ for $i=1,2, j=0,1,2$, and therefore $\gamma = V$.

In order to compute $V$, observe that it  is the unique polynomial $V$ of degree $< \deg(u u^{\varphi}) = 4$ such that
$V \equiv v \bmod{u}$ and $V \equiv v^{\varphi} \bmod{u^{\varphi}}. $
Keeping in mind that $u, u^{\varphi}$ are coprime, and using the Chinese Remainder Theorem (or following the explicit formulas in \cite{lange-03}), we get
$$ V = su+v \quad \text{ where } \quad s \equiv (v^{\varphi}-v)/u \bmod {u^{\varphi}}, $$
as claimed.
\end{proof}

Denoting $u(x) = x^2 + u_1x + u_0$ and $v(x) = v_1x + v_0$, we compute the compression $(\beta_0,\gamma_0,\gamma_1,\gamma_2,1)$ of $D$ according to the following formulas. We abbreviate $$ U_0 = u_0-u_0^q,\;\; U_1 = u_1-u_1^q,\;\; V_0 = v_0 - v_0^q,\;\; V_1 = v_1-v_1^q. $$
Then
\begin{eqnarray*}
 d & = & (U_1V_0 - U_0V_1)^{-1} \\
 \beta_0 & = & ((u_0u_1^q-u_0^qu_1)U_1 - U_0^2)d \\
 \gamma_0 & = & ((u_0v_0^q - u_0^qv_0)U_0 + (u_0^qu_1v_0-u_0u_1^qv_0^q - u_0^{q+1}V_1)U_1)d \\
 \gamma_1 & = & ((u_0v_1^q-u_0^qv_1)U_0 + (u_1^qv_0+u_0^qv_1^q)u_1U_1 + (u_0^qu_1 - u_0u_1^q)V_0 + (u_0v_1+u_1v_0^q)(u_1^{2q}-u_1^{q+1}))d \\
 \gamma_2 & = & (((u_1+u_1^q)U_1-U_0)V_0 - (u_0u_1-u_0^qu_1^q)V_1)d.
\end{eqnarray*}
Computing these values in the straightforward way takes 2S+32M+1I in $\F_{q^3}$. This number could probably be optimized by regrouping the terms in a more sophisticated way. 

\medskip \noindent {\bf Decompression.} Since decompression is dominated by factoring polynomials, we do not perform an exact operation count here. The algorithm computes
\begin{eqnarray*}
 S_1 & = & -2\gamma_2 + \beta_0^2\\
 S_2 & = & 2\gamma_1 + \gamma_2^2\\
 S_3 & = & -2\gamma_0-2\gamma_1 \gamma_2 + \beta_0^2f_3\\
 S_4 & = & 2 \gamma_0 \gamma_2 + \gamma_1^2 - \beta_0^2f_2\\
 S_5 & = & -2 \gamma_0 \gamma_1 + \beta_0^2f_1\\
 S_6 & = & \gamma_0^2-\beta_0^2f_0
\end{eqnarray*}
over $\Fq$ to obtain $H_D = x^6 - S_1x^5 + S_2x^4 - S_3x^3 + S_4x^2 - S_5x + S_6$. In almost all cases we are decompressing a divisor of the shape that we consider above for compression. $H_D$ either splits over $\Fq$ into two factors of degree 3, or it is irreducible over $\Fq$. Factoring $H_D$ over $\Fq$ takes $O(\log q)$ operations in $\Fq$. Then we factor either two polynomials of degree 3 over $\F_{q^3}$, or one degree 6 polynomial over $\F_{q^3}$, in $O(\log q)$ operations in $\F_{q^3}$. In all cases, we then compute the corresponding $v$-polynomials. It follows that the overall complexity of decompression is $O(\log q)$ operations in $\Fq$.

\end{document}